\documentclass[sigplan,screen,nonacm]{acmart}
\startPage{1}
\setcitestyle{nosort}

\usepackage[utf8]{inputenc}
\usepackage{import,todonotes,thm-restate}

\newif\ifdraft\draftfalse

\newcommand{\Cc}{\mathcal{C}}
\newcommand{\Ee}{\mathcal{E}}
\newcommand{\Ff}{\mathcal{F}}

\newcommand{\Jj}{\mathcal{J}}

\newcommand{\Ll}{\mathcal{L}}
\newcommand{\Oo}{\mathcal{O}}
\newcommand{\Pp}{\mathcal{P}}

\newcommand{\Nat}{\mathbb{N}}
\newcommand{\set}[1]{\{#1\}}
\newcommand{\class}[1]{[#1]}
\newcommand{\quotient}[1]{[#1]}
\newcommand{\eats}[1]{\mathord{|#1\rangle}}
\newcommand{\back}{\mathit{back}}
\newcommand{\Act}{\mathbb{A}}

\renewcommand{\exp}{\mathsf{exp}}
\newcommand{\RegStr}{\textsc{RegStr}}
\newcommand{\dist}{\textsc{Dist}}
\newcommand{\sgrow}{\textsc{Stack-Growth}}
\newcommand{\col}{\textsc{Color}}
\newcommand{\dplus}{d^{+\!}}
\newcommand{\Nabla}[1]{\nabla_{\!#1}}

\renewcommand{\epsilon}{\varepsilon}
\renewcommand{\rho}{\varrho}

\ifdraft

\newcommand{\pp}[1]{\todo[inline,size=\scriptsize,backgroundcolor=ACMLightBlue]{#1 - \textbf{Paweł}}}
\newcommand\sg[1]{\todo[inline,size=\scriptsize,backgroundcolor=yellow]{#1 - \textbf{Stefan}}}
\newcommand{\ppchanged}[1]{{\color{cyan}{#1}}}

\else

\newcommand\pp[1]{}
\newcommand\sg[1]{}

\newcommand\ppchanged[1]{#1}

\fi

\newcommand{\noQED}{\renewcommand\qedsymbol{}}
\newcommand{\restoreQED}{\renewcommand\qedsymbol{$\square$}}
\newcommand{\restore}[1]{\noQED#1*\restoreQED}

\begin{document}

\theoremstyle{acmplain}
\newtheorem{observation}[theorem]{Observation}
\theoremstyle{acmdefinition}
\newtheorem*{remark*}{Remark}

\title{Bisimulation Finiteness of Pushdown Systems Is Elementary}

\author{Stefan G\"oller}
\affiliation{%
  \institution{University of Kassel}
  \department{School of Electrical Engineering and Computer Science}
  \city{Kassel}
  \country{Germany}}
\email{stefan.goeller@uni-kassel.de}

\author{Paweł Parys}
\orcid{0000-0001-7247-1408}
\affiliation{%
  \institution{University of Warsaw}
  \department{Institute of Informatics}
  \city{Warsaw}
  \country{Poland}}
\email{parys@mimuw.edu.pl}

\begin{abstract}
	We show that in case
	a pushdown system is bisimulation equivalent to a finite system,
	there is already a bisimulation equivalent finite system whose size is elementarily bounded in the description size of the pushdown system.
	As a consequence we obtain that it is elementarily decidable if a given pushdown system is bisimulation equivalent to some finite system.
	This improves a previously best-known \textsf{ACKERMANN} upper bound for this problem.
\end{abstract}



\maketitle

\section{Introduction}

\paragraph{General background.}

The {\em class membership problem} for a language class $\mathcal{C}$ and a subclass $\mathcal{C}'\subseteq\mathcal{C}$
asks, given some device (like an automaton or a grammar) accepting a language in $\mathcal{C}$
to decide if the language that is described by the device is even a member
of $\mathcal{C}'$.
A prominent and particular such class membership problem is the {\em regularity problem},
i.e. the class
$\mathcal{C}'$ is the class of regular languages.
It is well-known that the regularity problem for (the class $\mathcal{C}$ of)
context-free languages is undecidable, see Hunt~\cite{Hunt82} for a general approach.
For deterministic context-free languages Stearns~\cite{Stearns67} showed a triply exponential upper bound
that was shortly thereafter improved to a doubly exponential upper bound by Valiant~\cite{Valiant75}.
It is fair to mention that the computational complexity of the regularity problem for deterministic context-free languages
is not yet well understood to date: to the best of the authors' knowledge the problem lies is $2$-$\mathsf{EXPTIME}$ (deterministic double exponential time) and is hard for $\mathsf{P}$,
a complexity gap that has been prevailing for around 45 years!
Similar large complexity gaps exist for various class membership problems for visibly
pushdown languages~\cite{BLS06}.

In the context of formal verification and concurrency theory the central notion of
equivalence is bisimulation equivalence
\cite{vanGlabbeek90},
which refines trace/language equivalence.
Indeed, the bisimulation-invariant fragment of monadic second-or\-der logic and various sublogics thereof have elegant
characterizations in
terms of prominent temporal logics \cite{JW96,ESV15,JvB1976,MR03,Carreiro15}.
In the context of verification and concurrency the regularity problem thus becomes the
{\em bisimulation finiteness problem},
which asks whether
the (finitely presented) input infinite transition system is
bisimulation equivalent
({\em bisimilar} for short) to some
finite-state system.
Decidability and complexity results for the bisimulation finiteness problem are known for only a few classes of infinite-state systems, see Srba~\cite{Srba04} for an overview.
Moreover, in case when decidability is known oftentimes large complexity gaps remain, the class of one-counter
systems being an exception, for which the problem is
$\mathsf{P}$-complete~\cite{BGJ14}.

\paragraph{Bisimulation finiteness of pushdown systems.}

Only recently Jan\v{c}ar~\cite{Jancar-arxiv,Jancar16,JANCAR2019} proved decidability of the bisimulation finiteness problem for pushdown systems,
a central class of infinite-state systems with a decidable monadic second-order theory~\cite{MS85}
and for which formal verification tools have been
developed~\cite{Moped,JMoped,TMP09,RLK07,RSJM05}.
It is worth mentioning that only a slight modification of the bisimulation finiteness problem, namely the question if there exists a
reachable configuration that is bisimulation-finite, becomes undecidable for order-two pushdown systems~\cite{BG12}.
Jan\v{c}ar's decidability result even holds for pushdown systems with deterministic $\varepsilon$-popping
rules~\cite{Jancar16,JANCAR2019}, a class
of graphs/systems also known as equational graphs \cite{DBLP:books/el/leeuwen90/Courcelle90} of finite out-degree.
Semi-decidability of the problem has long been known by simply enumerating all finite systems and checking whether
some of these finite systems is bisimilar to the input pushdown system.
As a side remark, it is worth mentioning that the bisimulation finiteness problem should not be confused with the question if a
given pushdown is bisimilar to a given finite system, the latter problem being $\mathsf{PSPACE}$-complete as shown
by Ku\v{c}era and Mayr~\cite{KM10}.
The difficult component of Jan\v{c}ar's decidability proof is to establish semi-decidability of the complement,
that is, to provide a semi-decision procedure that
halts in case a given pushdown system
is not bisimilar to a finite system.
A central ingredient to this semi-decision procedure
 is an oracle call to test the equivalence
of pushdown systems, the latter itself being an intricate problem whose decidability has been proven by S{\'e}nizergues
\cite{Senizergues05}. Only recently an Ackermannian upper bound for bisimilarity of
pushdown system has been proven by Jan\v{c}ar and Schmitz~\cite{JS19};
\textsf{ACKERMANN}-hardness is only known to hold in the presence of deterministic $\varepsilon$-popping rules \cite{Jancar13},
 whereas without $\varepsilon$-rules the problem is nonelementary~\cite{BGKM13}.
Coming back to bisimulation finiteness of pushdown systems, the oracle calls to a bisimulation equivalence check for pushdown systems is the inherent bottleneck of Jan\v{c}ar's approach.
The approach therefore contains a nonelementary complexity bottleneck and only guarantees
an \textsf{ACKERMANN} upper bound.
This stands in stark contrast to the best-known lower bound for bisimulation finiteness
of pushdown systems, namely $\mathsf{EXPTIME}$-hardness~\cite{KM02,Srba02}.

\paragraph{Our contribution.}

In this paper we prove that in case a pushdown system is bisimulation-finite
there is already a finite-state system whose size only elementarily depends
on the description size of the pushdown system, or equivalently,
a bisimulation-finite pushdown system only contains an elementary number
of bisimulation classes.
Using the above-mentioned polynomial space procedure for checking if a pushdown
system is bisimilar to a given finite-state system this implies an elementary decision
procedure for the bisimulation finiteness problem for pushdown systems.
Our approach avoids oracle calls for the equivalence problem for pushdown systems.
We follow a general proof strategy by Jan\v{c}ar in the aspect that we compare
configurations of the form $q\alpha\beta^e\gamma$ for large powers $e$ against
their infinite
approximants $q\alpha\beta^\omega$.
The core of our
proof (Section \ref{sec:core}) is to establish
the impossibility of a situation
that certain configurations
of the form $q\alpha\beta^e\gamma$ for a sufficiently large (elementary) $e$ are bisimilar
to configurations
of the form $q\alpha\beta^\omega$ unless the system is
bisimulation-finite.

\paragraph{\ppchanged{Future research questions.}}

Our contribution leads to further research questions.
It seems worth investigating if our approach can be
applied to the bisimulation finiteness problem for
further classes of
finitely-branching infinite-state systems that enjoy a pumping lemma property,
and for which neither the bisimulation
finiteness problem nor the bisimulation equivalence problem are known to be decidable
(or the complexity gap is extremely large),
for instance for ground tree rewrite systems~\cite{Loeding2003},
PAD and PA processes~\cite{Mayr98},
higher-order pushdown systems,
and of course for the aforementioned extension of pushdown systems with
deterministic $\varepsilon$-popping rules.
For BPA processes
(which are nothing but pushdown systems over a singleton set of control
states)
it seems interesting to find out if our
approach can be used to close an exponential complexity gap for bisimulation finiteness,
see Srba~\cite{Srba04} for an overview.
Finally, it seems worth investigating if
our technique
may lead to a potential future line of attack for
the equivalence problem of (deterministic)
pushdown systems.

\section{Preliminaries}\label{sec:preliminaries}

For any alphabet $A$ we denote by $A^*$ (resp. $A^+$) the set
of finite words (resp. the set of non-empty finite words) over $A$.
By $\Nat=\set{0,1,2,\ldots}$ we denote the set of the non-negative integers.
For any set $U\subseteq\Nat$ we denote by
$\min U$ the minimal element of the set $U$ with the convention that
$\min \emptyset=\omega$.
For $i,j\in\Nat$ we define $[i,j]=\set{k\in\Nat\mid i\leq k\leq j}$.

Let $f\colon\Nat^k\rightarrow\Nat$ be a function.
We say that $f$ is {\em elementary}
if $f$ is obtained by a composition of the following functions:
constant functions, projection, addition, multiplication and exponentiation.
For instance $(m,n)\mapsto2^{m^{n^2}}$ is an elementary function.
We write $f(\vec{n})\leq\exp(\vec{n})$ if there exists a
polynomial $p\colon\Nat^k\rightarrow\Nat$ such that $f(\vec{n})\leq 2^{p(\vec{n})}$
for all $\vec{n}\in\Nat^k$.
All functions $f\colon\Nat^k\rightarrow\Nat$ in this paper are elementarily upper bounded,
that is, there exists an elementary function $g\colon\Nat^k\rightarrow\Nat$ such
that $f(\vec{n})\leq g(\vec{n})$ for all $\vec{n}\in\Nat^k$.

A \emph{labeled transition system (LTS)} is a tuple
$\Ll=(S,\Act,\allowbreak(\rightarrow_a\nobreak)_{a\in\Act})$,
where $S$ is a (possibly infinite) set of {\em states},
$\Act$ is a finite set of {\em action symbols},
$(\rightarrow_a)\subseteq S\times S$ is a binary relation for
all $a\in\Act$.
We say $\Ll$ is {\em finite} if $S$ is finite.
We define its {\em size} as $|\Ll|=|S|$, thus $|\Ll|\in\Nat$ if $\Ll$ is finite
and $|\Ll|=\omega$ if not.
We define the binary relation $(\rightarrow)=\bigcup_{a\in\Act}(\rightarrow_a$).
For all $s,t\in S$ we define
$\dist(s,t)=\min\set{m\in\Nat\mid s\rightarrow^m t}\in\Nat\cup\set{\omega}$,
the length of the shortest path from $s$ to $t$ in $\Ll$.

For such an LTS $\Ll$ we say a binary relation $R\subseteq S\times S$ is a {\em bisimulation}
if the following back-and-forth property holds
for all $a\in\Act$ and all $(s,t)\in R$:
for all $s\rightarrow_a s'$ there exists some $t\rightarrow_a t'$ such that
$(s',t')\in R$ and, conversely, for all
$t\rightarrow_a t'$ there exists some $s\rightarrow_a s'$ such that
$(s',t')\in R$.
Observe that the union of two bisimulations is again a bisimulation.
We write $s\sim t$ if $(s,t)\in R$ for some bisimulation relation $R$;
note that $(\sim)\subseteq S\times S$ is the largest bisimulation on $S$.
If $s\sim t$, we say that $s$ and $t$ are \emph{bisimilar}.
For every state $s\in S$ we denote by $\class{s}=\set{t\in S\mid s\sim t}$
the {\em bisimulation class of $s$}.
The {\em bisimulation quotient} $\quotient{\Ll}$ is the LTS
$\quotient{\Ll}=(\set{\class{s}\mid s\in S},\Act,(\rightarrow_a')_{a\in\Act})$,
where $c\rightarrow_a' d$ if $s\rightarrow_a t$ for some $s\in c$, $t\in d$.

A {\em pushdown system (PDS)} is a tuple $\Pp=(Q,\Gamma,\Act,\Delta)$,
where $Q$ is a finite set of {\em control states},
$\Gamma$ is a finite set of {\em stack symbols},
$\Act$ is a finite set of {\em action symbols},
and $\Delta\subseteq Q\times\Gamma\times\Act\times Q\times\Gamma^*$ is a
finite transition relation.
In this paper, elements of $\Gamma$ are typically denoted in capital letters,
such as $X,Y,Z\ldots$, whereas elements of $\Gamma^*$ (i.e.,
finite sequences over $\Gamma$) are typically denoted by small Greek
letters such as $\alpha,\beta,\gamma,\ldots$.
If $|\alpha|\leq 2$ for all transitions $(p,X,a,q,\alpha)\in\Delta$,
then we say that $\Pp$ is in a
{\em push-pop normal form}.
The {\em size} $|\Pp|$ of $\Pp$ is defined as $|\Pp|=|Q|+|\Gamma|+|\Act|+|\Delta|$.

Elements of $\Gamma^*$ are {\em stack contents} and elements of $Q\times\Gamma^*$ are {\em configurations} of $\Pp$.
For every $\delta=(p,X,a,q,\beta)\in\Delta$ we define the binary relation
$(\xrightarrow{\delta})=\set{(p X\alpha,q\beta\alpha)\mid \alpha\in\Gamma^*}$.
The relation $\xrightarrow{\rho}$ is naturally extended to
finite sequences $\rho\in\Delta^*$.
A {\em run} from a source configuration $p\alpha$ to a target
configuration $q\beta$
is a sequence $\rho\in\Delta^*$ such that
$p\alpha\xrightarrow{\rho}q\beta$.
Note that a run can be a run for numerous pairs of source and target
configurations.
When $\rho=\delta_1\cdots\delta_n\in\Delta^n$ and $0\leq i\leq j\leq n$, we write $\rho[i,j]$ for the \emph{subrun} $\delta_{i+1}\dots\delta_j$.
Slightly abusing notation, when the starting configuration
$p\alpha$ of such a run is fixed
from the context
we sometimes prefer to write $\rho(i)$ to denote the unique configuration
$q_i\beta_i$ that satisfies $p\alpha\xrightarrow{\rho[0,i]}q_i\beta_i$; in particular $\rho(0)=p\alpha$.

A PDS $\Pp=(Q,\Gamma,\Act,\Delta)$ together with an \emph{initial configuration} $s_0\in Q\times\Gamma^*$ induces an infinite
LTS $\Ll(\Pp,s_0)=(S,\Act,\allowbreak(\rightarrow_a\nobreak)_{a\in\Act})$,
where $S=\{q\alpha\in Q\times\Gamma\mid s_0\rightarrow^* q\alpha\}$ is the set of
configurations reachable from $s_0$,
and $(\rightarrow_a)=\set{(p X\alpha,q\beta\alpha)\in S\times S\mid
(p,X,a,q,\beta)\in\nobreak\Delta}$ for all $a\in\Act$.

The set $\RegStr(\Gamma)$ contains all elements of $\Gamma^*$, and all infinite strings of the form $\alpha\beta\beta\beta\dots$,
denoted $\alpha\beta^\omega$ (and called \emph{$\omega$-approximants}), where $\alpha\in\Gamma^*$ and $\beta\in\Gamma^+$.
As in Jančar~\cite{Jancar-arxiv,JANCAR2019}, we sometimes prefer to extend the
stack content of configurations of pushdown systems from $\Gamma^*$
to $\RegStr(\Gamma)$.
All of the above notions are analogously defined for
pushdown systems whose configurations are in $\RegStr(\Gamma)$.

We are now ready to state the main result of this paper.
\begin{restatable}{theorem}{ThmMain}\label{thm:main}
	There is an elementary function $\varphi\colon\Nat^2\to\Nat$ such that if $\quotient{\Ll(\Pp,p_0\alpha_0)}$ is finite for some initial configuration $p_0\alpha_0$,
	then $|\quotient{\Ll(\Pp,p_0\alpha_0)}|\leq\varphi(|\Pp|,|\alpha_0|)$.\qed
\end{restatable}

The following corollary is an immediate consequence of Theorem~\ref{thm:main}
and a result by Ku\v{c}era and Mayr \cite{KM10},
namely that checking whether
a given pushdown system is
bisimilar to a given finite system is $\mathsf{PSPACE}$-complete.

\begin{corollary}
	Given a PDS $\Pp$ and a configuration $p_0\alpha_0$ of $\Pp$, the question whether
	$\quotient{\Ll(\Pp,p_0\alpha_0)}$ is finite is elementarily decidable.
\qed\end{corollary}

\paragraph{Convention.}

Since we can replace every transition pushing
a sequence onto the stack by a transition that pushes a suitable fresh symbol that represents such a
sequence,
it is clear that for every PDS $\Pp$
one can compute
in polynomial time a bisimilar PDS $\Pp'$ that
is in a push-pop normal form
(we refer to Lemma~\ref{lem:arbitrary2PushPop}
in the Appendix for more details).
Thus, towards proving Theorem~\ref{thm:main} we fix, for the rest of this paper, a PDS
$\Pp=(Q,\Gamma,\Act,\Delta)$ that is in a push-pop normal form.

\section{Some basics on pushdown systems}\label{sec:basics}

In this section we present some definitions and known facts about pushdown systems, useful in our proofs.
Proofs of these facts can be found in the Appendix.

For a stack content $\alpha\in\Gamma^*$, let $\eats{\alpha}$ be a function that maps every set of control states $T$ to the set $\set{r\in\nobreak Q\mid\allowbreak \exists q\in T.\,q\alpha\to^*r}$.
Observe that $\eats{\beta}(\eats{\alpha}(T))=\eats{\alpha\beta}(T)$
for all $\alpha,\beta\in\Gamma^*$ and all $T\subseteq Q$.
The following lemma is a direct consequence of the definition.

\begin{restatable}{lemma}{LemExtendEats}\label{lem:extend-eats}
	Let $p,q\in Q$ and $\alpha,\beta\in\Gamma^*$.
	If $p\alpha\to^*q\beta$, then $\eats{\alpha}(p)\supseteq\eats{\beta}(q)$.
\qed\end{restatable}

For a transition $\delta=(p,X,a,q,\alpha)$ let $\sgrow(\delta)=|\alpha|-1$, and for a run $\rho=\delta_1\dots\delta_n\in\Delta^n$ let 
\begin{align*}
	\sgrow(\rho)=\sum_{i=1}^n\sgrow(\delta_i)\,.
\end{align*}
In other words, if $\rho$ is a run from $p\alpha$ to $q\beta$ (with $\alpha,\beta\in\Gamma^*$), then $\sgrow(\rho)=|\beta|-|\alpha|$.

A run $\rho$ is called {\em augmenting} if $\sgrow(\rho[0,i])\geq 0$ for all $i\in[0,n]$.
Thus, an augmenting run from $pX\!\alpha$ does not ``dig into'' the stack content $\alpha$,
but rather all its configurations are of the form $p_i\kappa_iY\!\alpha$ for some $\kappa_i\in\Gamma^*$.

The following simple lemma states that there are
at most exponentially (in $z$) many configurations reachable by
augmenting runs of stack growth at most $z$.

\begin{restatable}{lemma}{LemRadius}\label{lem:radius}
	If $p\alpha$ is a configuration of $\Pp$, and $z\in\Nat$,
	then there are at most $|\Pp|^{z+2}$ configurations $q\beta$
	such that $p\alpha\xrightarrow{\rho}q\beta$ for some augmenting run $\rho$ satisfying $\sgrow(\rho)\leq z$.
\qed\end{restatable}

Our next lemma says that if there is a run between two similar configurations, then there is a short run between them.
Essentially, this boils down to the standard pumping lemma for pushdown automata.

\begin{restatable}{lemma}{LemPDSPaths}\label{lem:PDS-paths}
	There exists a constant $\Ee\leq\exp(|\Pp|)$ such
	that whenever $p\alpha\rightarrow^*q\beta$ for two configurations $p\alpha, q\beta$ of $\Pp$, then $\dist(p\alpha,q\beta)\leq(|\alpha|+|\beta|)\cdot\Ee$.
\qed\end{restatable}

In our proof we depend on the notion of linked pairs.
A similar notion appeared in Jančar~\cite{Jancar-arxiv,JANCAR2019}.

\begin{definition}
	A pair $(\alpha,\beta)\in\Gamma^*\times\Gamma^+$ is a \emph{linked pair}
	if $\eats{\alpha}=\eats{\alpha\beta}$ and $\eats{\beta}=\eats{\beta\beta}$.
\end{definition}

It is often easier to apply any kind of a pumping argument to a linked pair than to an arbitrary stack content.
Simultaneously, a linked pair can be found on top of every sufficiently large
stack content, as described by the next lemma.

\begin{restatable}{lemma}{LemDecompose}\label{lem:decompose}
	There is a constant $\Ff\leq\exp(\exp(|\Pp|))$ such that
	every configuration $q\delta$ of $\Pp$ reachable from an initial configuration $q_0\alpha_0$ with $|\delta|\geq \Ff+|\alpha_0|$
	can be written as $q\delta=q\alpha\beta\gamma$,
	where $(\alpha,\beta)$ is a linked pair,
	$|\alpha|,|\beta|\leq\Ff$, and all configurations of the form $q\alpha\beta^i\gamma$ (where $i\in\Nat$) are reachable from $q_0\alpha_0$.
\qed\end{restatable}

\begin{remark*}
	The authors are convinced that the doubly-ex\-po\-nen\-tial upper bound in Lemma~\ref{lem:decompose} is optimal.
	On the other hand, it seems quite possible that one can replace the notion of 
	linked pairs by some weaker notion,
	so that an analogue of Lemma~\ref{lem:decompose} would give an exponential upper bound on the length of the fragments $\alpha$ and $\beta$,
	and in effect the complexity of the whole algorithm would be decreased exponentially.
	This possibility was not investigated by the authors.
\end{remark*}

Below we also state an easy lemma, which follows directly from definitions.

\begin{restatable}{lemma}{LemEatsBeta}\label{lem:eats-beta}
	Let $r\in Q$, and let $(\alpha,\beta)$ be a linked pair.
	Then for every $r'\in\eats{\alpha}(r)$ we have
	$\eats{\beta}(r')\subseteq\eats{\alpha}(r)$.
\qed\end{restatable}

Observe that if $(\alpha,\beta)$ is a linked pair, then $(\beta,\beta)$ is a 
linked pair as well.
In consequence, the above lemma (as well as Corollary~\ref{cor:repetition-is-omega} stated below) can be used when replacing $\alpha$ by $\beta$.

For the rest of this section we state two
straightforward lemmata, which are useful while proving that two configurations are
bisimilar. Both
lemmata
have appeared in a related form already in Jančar~\cite{Jancar-arxiv,JANCAR2019}.
The first of them talks about a situation when we add something on top of
bisimilar configurations.

\begin{restatable}{lemma}{LemEquivBecauseEquivBelow}\label{lem:equiv-because-equiv-below}
	Let $q\in Q$, and $\alpha\in\Gamma^*$, and $\gamma,\gamma'\in\RegStr(\Gamma)$.
	If $r\gamma\sim r\gamma'$ for every $r\in\eats{\alpha}(q)$, then $q\alpha\gamma\sim q\alpha\gamma'$.
\qed\end{restatable}

The second lemma states that if we detect a loop while popping a stack $\beta$, then this loop can be repeated forever,
and thus the stack content is equivalent to $\beta^\omega$.

\begin{restatable}{lemma}{LemRepetitionIsOmega}\label{lem:repetition-is-omega}
	Let $U\subseteq Q$, $\beta\in\Gamma^+$, $\gamma\in\RegStr(\Gamma)$, and $i,j\in\Nat$.
	If $\eats{\beta}(U)\subseteq U$, and $r\beta^i\gamma\sim r\beta^j\gamma$ for all $r\in U$, and $i\neq j$,
	then actually $r\beta^i\gamma\sim r\beta^j\gamma\sim r\beta^\omega$ for all $r\in U$.
\qed\end{restatable}

We now combine the above two lemmata into a single corollary.

\begin{restatable}{corollary}{CorRepetitionIsOmega}\label{cor:repetition-is-omega}
	Let $q\in Q$, let $(\alpha,\beta)$ be a linked pair, let $\gamma,\gamma'\in\RegStr(\Gamma)$, and let $i,j\in\Nat$.
	If $r\beta^i\gamma\sim r\beta^j\gamma$ for all $r\in\eats{\alpha}(q)$, and $i\neq j$,
	then $q\alpha\beta^i\gamma\sim q\alpha\beta^j\gamma\sim q\alpha\beta^\omega$.
\qed\end{restatable}

\section{The heart of the proof}

In addition to all the lemmata of Section~\ref{sec:basics}, 
we give here a simple but decisive lemma (Lemma~\ref{lem:equiv-because-loop}), 
which is central for our proof.
It assumes a situation in which several bisimulation classes are determined by some
(small) stack contents and certain bisimulation classes one reaches when popping these
stack contents; moreover all of these latter bisimulation classes in turn are themselves determined by a
(small) stack content below which we see bisimulation classes for which we have already made this characterization!
More precisely, while
popping from a configuration bisimilar to $p\mu_1$ we reach a configuration $r\nu_1$,
and while popping from a configuration bisimilar to $r\nu_1$ we reach again $p\mu_1$.
In such a circular situation we can replace $\mu_1$ and $\nu_1$ by other stack contents having the same property.

\begin{restatable}{lemma}{LemEquivBecauseLoop}\label{lem:equiv-because-loop}
	Let $U,V\subseteq Q$ and assume that for all $p\in U$ there
	are $r_p\in Q$ and $\chi_p\in\Gamma^+$ such that
	$\eats{\chi_p}(r_p)\subseteq V$ and conversely
	for all $r\in V$ there are $p_r\in Q$ and $\xi_r\in\Gamma^+$
	such that $\eats{\xi_r}(p_r)\subseteq U$.
	If for some
	$\mu_1,\mu_2,\nu_1,\nu_2\in\RegStr(\Gamma)$
	we have for all $i\in\set{1,2}$,
	\begin{align*}
		p\mu_i&\sim r_p\chi_p\nu_i&&\text{for all }p\in U
		&&\text{and}\\
		r\nu_i&\sim p_r\xi_r\mu_i&&\text{for all }r\in V,
	\end{align*}
	then $p\mu_1\sim p\mu_2$ for all $p\in U$ and
	$r\nu_1\sim r\nu_2$ for all $r\in V$.
\end{restatable}

\begin{proof}
	We define a relation $R$ between configurations (we are going to prove that this relation is a bisimulation):
	\begin{align*}
		R={}&\set{(s_1,s_2)\mid\exists q\in Q.\,\exists \delta\in\Gamma^*.\,\\
			&\hspace{5em}s_1\sim q\delta\mu_1\land
			q\delta\mu_2\sim s_2\land
			\eats{\delta}(q)\subseteq U}\ \cup\\
		&\set{(s_1,s_2)\mid\exists q\in Q.\,\exists \delta\in\Gamma^*.\,\\
			&\hspace{5em}s_1\sim q\delta\nu_1\land
			q\delta\nu_2\sim s_2\land
			\eats{\delta}(q)\subseteq V}\,.
	\end{align*}
	Observe that in particular $(p\mu_1,p\mu_2)\in R$ for all $p\in U$ and $(r\nu_1,r\nu_2)\in R$ for all $r\in V$,
	which implies the thesis if $R$ is a bisimulation.

	In order to prove that $R$ is a bisimulation, consider a pair $(s_1,s_2)\in R$.
	We have four possible reasons for $(s_1,s_2)\in R$.
	\begin{enumerate}
	\item
		One possibility is that $s_1\sim qX\eta\mu_1$ and $qX\eta\mu_2\sim s_2$
		for some control state $q\in Q$, some stack symbol $X\in\Gamma$, and some stack content $\eta\in\Gamma^*$, where $\eats{X\eta}(q)\subseteq U$
		(we have replaced $\delta$ from the definition of $R$ by $X\eta$, assuming that $|\delta|>0$).

		Suppose that $s_1\to_a t_1$ for some configuration $t_1$; we should prove
		the existence of a configuration $t_2$ such that $s_2\to_a t_2$ and $(t_1,t_2)\in R$.
		Because $s_1\sim qX\eta\mu_1$, there is $t_1'$ such that $qX\eta\mu_1\to_a t_1'$ and $t_1\sim t_1'$.
		Necessarily $t_1'=q'\alpha\eta\mu_1$ for some transition $(q,X,a,q',\alpha)\in\Delta$.
		Due to the same transition we have that $qX\eta\mu_2\to_a q'\alpha\eta\mu_2$,
		and because $qX\eta\mu_2\sim s_2$, there is a configuration $t_2$ such that $s_2\to_a t_2$ and $q'\alpha\eta\mu_2\sim t_2$.
		Moreover, because $qX\eta\to^*q'\alpha\eta$, by Lemma~\ref{lem:extend-eats} we have that $\eats{\alpha\eta}(q')\subseteq\eats{X\eta}(q)\subseteq U$, so $(t_1,t_2)\in R$.

		Likewise, whenever $s_2\to_a t_2$ for a configuration $t_2$, we should
		prove the existence of a configuration $t_1$
		such that $s_1\to_a t_1$ and $(t_1,t_2)\in R$.
		This is completely symmetric to what we have done above.

	\item	Another reason for $(s_1,s_2)\in R$ is that $s_1\sim qX\eta\nu_1$ and $qX\eta\nu_2\sim s_2$
		for some control state $q\in Q$, some stack symbol $X\in\Gamma$, and some stack content $\eta\in\Gamma^*$, where $\eats{X\eta}(q)\subseteq V$.
		This case is completely symmetric to the previous one, thus we do not repeat here the proof.

	\item	It is also possible that $s_1\sim q\mu_1$ and $q\mu_2\sim s_2$ for some control state $q\in Q$ such that $\eats{\epsilon}(q)\subseteq U$
		(i.e., when $\delta$ from the definition of $R$ is empty).
		The condition $\eats{\epsilon}(q)\subseteq U$ says simply that $q\in U$.
		Thus, by assumptions of the lemma we have that
		$s_1\sim q\mu_1\sim r_q\chi_q\nu_1$, and $r_q\chi_q\nu_2\sim q\mu_2\sim s_2$, and $\eats{\chi_q}(r_q)\subseteq V$.
			Thus, we have reduced this case to Case 2 (taking $r_q$ as $q$ and 
			$\chi_q\in\Gamma^+$ as $\delta$).

	\item	The remaining case is that $s_1\sim q\nu_1$ and $q\nu_2\sim s_2$ for some control state $q\in Q$ such that $\eats{\epsilon}(q)\subseteq V$.
		Proceeding symmetrically to Case 3, this case can again be reduced to Case 1.
	\qedhere
	\end{enumerate}
\end{proof}

\section{Main technical theorem}

In this section we state our main technical theorem (Theorem~\ref{thm:technical}),
and we show how our main result (Theorem~\ref{thm:main}) follows from this theorem.

\begin{theorem}\label{thm:technical}
	There exists an elementary function $h\colon\Nat^3\rightarrow\Nat$
	such that for all $q\in Q$, all linked pairs $(\alpha,\beta)$,
	and all $\gamma\in\Gamma^*$, if
	there are only finitely many pairwise non-bisimilar
	configurations in $\set{q\alpha\beta^i\gamma\mid i\in\Nat}$, then
	$q\alpha\beta^e\gamma\sim q\ppchanged{\alpha}\beta^\omega$ for some $e\leq h(|\Pp|,|\alpha|,|\beta|)$.
\end{theorem}

\sg{$h$ is five-fold exponentially bounded.}

Heading towards proving Theorem~\ref{thm:main}, as a first step we observe
that Theorem~\ref{thm:technical} implies
that in case $\quotient{\Ll(\Pp,s_0)}$ is finite and $(\alpha,\beta)$ is a linked pair,
then there can only be an elementary number of bisimulation
classes among $\{q\alpha\beta\gamma\mid \gamma\in\Gamma^*\land\forall
i\in\Nat.\,s_0\to^*q\alpha\beta^i\gamma\}$.
This is formalized by the following lemma.

\begin{restatable}{lemma}{LemFour}\label{lem:4}
	There exists an elementary function $\lambda\colon\Nat^3\to\Nat$ such that for all $q\in Q$ and all 
	linked pairs $(\alpha,\beta)$,
	if $\quotient{\Ll(\Pp,s_0)}$ is finite for some initial configuration $s_0$,
	then there are at most $\lambda(|\Pp|,|\alpha|,|\beta|)$ pairwise non-bisimilar configurations in $\set{q\alpha\beta\gamma\mid\gamma\in\Gamma^*\land\allowbreak\forall i\in\Nat.\, s_0\to^*q\alpha\beta^i\gamma}$.
\end{restatable}

\sg{$\lambda$ is six-fold exponentially bounded.}

\begin{proof}
	Let us first sketch the proof.
	Clearly, for every configuration of the form $q\alpha\beta\gamma$
	that is in the set specified in Lemma~\ref{lem:4} we have
	$q\alpha\beta^e\gamma\sim q\alpha\beta^\omega$,
	where $e$ is an elementarily bounded number (cf.~Theorem~\ref{thm:technical}).
	Thus, every configuration $r\gamma$ with
	$r\in\eats{\alpha\beta^e}(q)=\eats{\alpha\beta}(q)$
	is reachable from $q\alpha\beta^e\gamma$ in at most $e'$ steps and thus bisimilar to a
	configuration reachable from $q\alpha\beta^\omega$ in at most $e'$ steps, for some
	elementary constant $e'$.
	Since moreover $\Ll(\Pp)$ has out-degree at most $|\Pp|$ there
	is only an elementary number of configurations in distance at most $e'$
	from $q\alpha\beta^\omega$.
	Hence, there are at most elementarily many tuples of bisimulation classes $[r\gamma]_{r\in\eats{\alpha\beta}(q)}$, where $\gamma$ ranges over all permissible stack contents.
	But since the bisimulation class
	of each such permissible $\gamma$ is determined by $q,\alpha,\beta$ and
	$[r\gamma]_{r\in\eats{\alpha\beta}(q)}$ the lemma follows.

	Coming to the details, we take
	\begin{align*}
		\lambda(|\Pp|,|\alpha|,|\beta|)&=(|\Pp|+1)^{(\lambda_1+1)\cdot|\Pp|}\,,&&\mbox{where}\\
		\lambda_1&=(|\alpha|+|\beta|\cdot h(|\Pp|,|\alpha|,|\beta|))\cdot\Ee\,,
	\end{align*}
	and where the constant $\Ee$ is taken from Lemma~\ref{lem:PDS-paths},
	and the function $h$ is taken from Theorem~\ref{thm:technical}.

	\sg{$\lambda$ is six-fold exponentially bounded.}

	Fix an initial configuration $s_0$, a control state $q\in Q$ and a linked pair $(\alpha, \beta)$.
	Let $\Omega$ be the set of those stack contents $\gamma\in\Gamma^*$ for which all configurations of the form $q\alpha\beta^i\gamma$ (where $i\in\Nat$) are reachable from $s_0$.
	To every $\gamma\in\Omega$ we assign a tuple of bisimulation classes $t_\gamma=(\class{r\gamma})_{r\in\eats{\alpha\beta}(q)}$.
	Observe that if $t_\gamma=t_{\gamma'}$ for some $\gamma,\gamma'\in\Omega$, then $q\alpha\beta\gamma\sim q\alpha\beta\gamma'$ by Lemma~\ref{lem:equiv-because-equiv-below}.
	Thus the maximal number of pairwise non-bisimilar configurations in $\set{q\alpha\beta\gamma\mid\gamma\in\Omega}$
	is bounded by the maximal number of different tuples $t_\gamma$.
	It remains to bound the latter.

	First, we claim that, for every $r\in\eats{\alpha\beta}(q)$ and for every $\gamma\in\Omega$, some configuration bisimilar to $r\gamma$ is reachable from $q\alpha\beta^\omega$ in at most $\lambda_1$ steps.
	Indeed, notice that there are only finitely many pairwise non-bisimilar configurations in $\set{q\alpha\beta^i\gamma\mid i\in\Nat}$
	(since $\gamma\in\Omega$, they all belong to $\Ll(\Pp,s_0)$, which is bisimulation-finite).
	Thus, by Theorem~\ref{thm:technical}, $q\alpha\beta^e\gamma\sim q\alpha\beta^\omega$ for some $e\leq h(|\Pp|,|\alpha|,|\beta|)$.
	Moreover, $q\alpha\beta^e\to^*r$ because $r\in\eats{\alpha\beta}(q)=\eats{\alpha\beta^e}(q)$ (the equality holds because $(\alpha,\beta)$ is a linked pair).
	By Lemma~\ref{lem:PDS-paths}, this implies that $\dist(q\alpha\beta^e,r)\leq(|\alpha|+|\beta|\cdot e)\cdot\Ee\leq\lambda_1$,
	hence also $\dist(q\alpha\beta^e\gamma,r\gamma)\leq\lambda_1$.
	Because $q\alpha\beta^e\gamma\sim\alpha\beta^\omega$, from $q\alpha\beta^\omega$ we can reach a configuration bisimilar to $r\gamma$ in the same number of steps,
	namely at most $\lambda_1$.

	Next, notice that every configuration has at most $|\Delta|\leq|\Pp|$ successors (by applying a particular transition to a particular configuration we obtain a particular successor).
	In effect, there are at most $(|\Pp|+1)^{\lambda_1+1}$ configurations reachable from $q\alpha\beta^\omega$ in at most $\lambda_1$ steps.
	Recalling that for all $r\in\eats{\alpha\beta}(q)$ and $\gamma\in\Omega$ a configuration bisimilar to $r\gamma$ is reachable from $q\alpha\beta^\omega$ in at most $\lambda_1$ steps,
	this means that there are at most $(|\Pp|+1)^{(\lambda_1+1)\cdot|Q|}\leq\lambda(|\Pp|,|\alpha|,|\beta)$ different tuples $t_\gamma=(\class{r\gamma})_{r\in\eats{\alpha\beta}(q)}$.
	This finishes the proof.
\end{proof}

Theorem~\ref{thm:main} is now easily shown.
Indeed, Lemma~\ref{lem:decompose} implies that
all but elementarily many small configurations
are in the set specified in Lemma~\ref{lem:4},
and thus they result in only elementarily many different bisimulation classes by Lemma~\ref{lem:4}.
A formal proof of this implication is given below.
\begin{figure*} 
	\begin{center}
		\import{pics/}{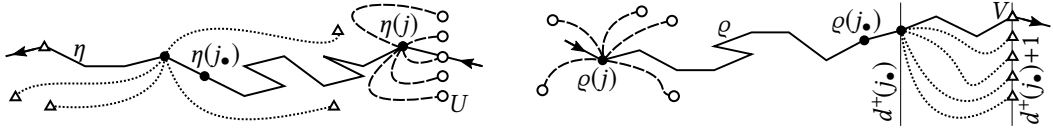}
	\end{center}
	\caption{A schematic illustration of runs $\eta$ and $\rho$ (stack grows to the left)}
	\label{fig:loop}
\end{figure*}

\begin{proof}[Proof of Theorem~\ref{thm:main}]
	Recall the constant $\Ff$ from Lemma~\ref{lem:decompose}.
	Let
	\begin{align*}
		S_{\mathit{small}}=\set{q\delta\in Q\times\Gamma^*\mid |\delta|<\Ff+|\alpha_0|}\,,
	\end{align*}
	and, for every control state $q\in Q$ and every linked pair $(\alpha,\beta)$ with $|\alpha|,|\beta|\leq\Ff$, let
	\begin{align*}
		S_{q,\alpha,\beta}=\set{q\alpha\beta\gamma\mid\gamma\in\Gamma^*\land\forall i\in\Nat.\, p_0\alpha_0\to^*q\alpha\beta^i\gamma}\,.
	\end{align*}

	Observe first that every configuration $q\delta$ of $\Pp$ reachable from $p_0\alpha_0$ belongs to one of the sets defined above.
	Indeed, if $|\delta|\geq\Ff+|\alpha_0|$ (i.e., if $q\delta\not\in S_{\mathit{small}}$),
	then by Lemma~\ref{lem:decompose} the configuration can be written as $q\delta=q\alpha\beta\gamma$, where $(\alpha,\beta)$ is a linked pair, $|\alpha|,|\beta|\leq\Ff$,
        and all configurations of the form $q\alpha\beta^i\gamma$ (where $i\in\Nat$) are reachable from $q_0\alpha_0$,
        meaning that $q\delta\in S_{q,\alpha,\beta}$.
        Thus, it remains to bound the maximal number of pairwise non-bisimilar configurations in the union of all the above sets.

	\sg{$\Ff$ is doubly exponentially bounded.}

	Clearly,
	\begin{align*}
		|S_{\mathit{small}}|\leq |Q|\cdot(|\Gamma|+1)^{\Ff+|\alpha_0|-1}\leq|\Pp|^{\Ff+|\alpha_0|}\,.
	\end{align*}
	Moreover, there are at most
	\begin{align*}
		|Q|\cdot(|\Gamma|+1)^\Ff\cdot(|\Gamma|+1)^\Ff\leq|\Pp|^{2\cdot\Ff+1}
	\end{align*}
	triples $(q,\alpha,\beta)$ where $q\in Q$, and $(\alpha,\beta)$ is a linked pair, and $|\alpha|,|\beta|\leq\Ff$.
	For every such a triple, by Lemma~\ref{lem:4}, there are at most $\lambda(|\Pp|,|\alpha|,|\beta|)\leq\lambda(|\Pp|,\Ff,\Ff)$ pairwise non-bisimilar configurations in $S_{q,\alpha,\beta}$.
	Thus, altogether in all the sets, there are at most
	\begin{align}
		\varphi(|\Pp|,|\alpha_0|)=|\Pp|^{\Ff+|\alpha_0|}+|\Pp|^{2\cdot\Ff+1}\cdot\lambda(|\Pp|,\Ff,\Ff)\label{eq:def-phi}
	\end{align}
	pairwise non-bisimilar configurations.
\end{proof}

\sg{$\varphi$ is six-fold exponentially bounded.}

\section{Overview of the proof\ppchanged{ of Theorem~\ref{thm:technical}}}\label{sec:overview}

This section is devoted to sketching the proof of Theorem~\ref{thm:technical}; more details on this proof are provided in subsequent sections.

Fix $\alpha,\beta,\gamma,q$ as in the statement of the theorem: $(\alpha,\beta)$ is a linked pair, 
and there are only finitely many pairwise non-bisimilar configurations in $\set{q\alpha\beta^i\gamma\mid i\in\Nat}$.
Let $V=\eats{\alpha}(q)$;
this is the set of control states reachable after popping the stack content $\alpha$ from the fixed control state $q$.
Because $(\alpha,\beta)$ is a linked pair, after popping a stack content of the form $\alpha\beta^i$, we can again only reach control states from $V$;
more precisely, $\eats{\alpha\beta^i}=\eats{\beta}(V)=V$, for all $i\in\Nat$.

For simplicity of the description, assume additionally in this section that $\eats{\beta}(r)=V$ for all $r\in V$,
that is, that from every control state $r$ in $V$ we can reach every other control state of $V$ (notice that if we can reach it after popping multiple copies of $\beta$,
then we can reach it also after popping a single copy of $\beta$, because $\eats{\beta\beta}=\eats{\beta}$).
This is not true in general, and causes some technical difficulties in the actual proof, presented in next sections;
in particular, in order to avoid this assumption, we need to consider so-called monochromatic intervals.

We now present particular steps of the proof.
Runs and configurations described below are depicted on Figure~\ref{fig:loop}.
\begin{enumerate}
\item
	Take the smallest number $e\in\Nat$ such that $q\alpha\beta^e\gamma\sim q\alpha\beta^\omega$.
	Assuming that $e$ is large enough (namely, larger than $h(|\Pp|,|\alpha|,|\beta|)$ for a function $h$ that is defined in the actual proof),
	we are heading towards a contradiction.
\item\label{pkt:2}
	We first observe that while ``going down'' from $q\alpha\beta^e\gamma$ we can visit many bisimulation classes.
	More precisely, to every number $d\in[0,s]$
	we can assign a tuple of bisimulation classes $(\class{r\beta^{e-d}\gamma})_{r\in V}$, reachable after popping $\alpha$ and $d$ copies of $\beta$ from $q\alpha\beta^e\gamma$.
	If the same tuple was assigned to two distinct values of $d$, say to $d_1$ and $d_2$,
	then by Corollary~\ref{cor:repetition-is-omega} we would have $q\alpha\beta^{e-{d_1}}\gamma\sim q\alpha\beta^{e-{d_2}}\gamma\sim q\alpha\beta^\omega$, contrarily to the minimality of $e$.
	It follows that among $r\alpha\beta^{e-d}\gamma$, for different $r$
	and $d$, we have a lot of pairwise non-bisimilar configurations.
\item\label{pkt:3}
	Next, we construct a run $\rho$ from $q\alpha\beta^e\gamma$ that ``quickly'' (as specified in Lemma~\ref{lem:PDS-paths}) pops $\alpha\beta^e$.
	Let $\Nabla{d}=r_d\beta^{e-d}\gamma$ be the configuration of $\rho$ at ``depth $d$'', that is, after popping $\alpha\beta^d$.
	Due to Point 2, we can ensure that among $\Nabla{d}$ there are many pairwise non-bisimilar configurations.
	Moreover, we can ensure this also locally: for every interval $[a,b]$, among $\Nabla{d}$ with $d\in[a,b]$ there are relatively many pairwise non-bisimilar configurations
	(where their number depends in an appropriate way on $b-a$).
	Points 2-3 are shown in Section~\ref{sec:go-down}.
\item\label{pkt:4}
	Since $q\alpha\beta^e\gamma\sim q\alpha\beta^\omega$, there must exist a run $\eta$ from $q\alpha\beta^\omega$ that mimics $\rho$, that is, such that $\rho(i)\sim\eta(i)$ for all positions $i$.
	Notice though, that after popping a copy of $\beta$ from $\beta^\omega$ we still have the same stack content $\beta^\omega$.
	In effect, $\eta$ cannot visit many pairwise non-bisimilar configurations by only
	popping something from $q\alpha\beta^\omega$; it has to push a lot.
	We may say that $\eta$ goes up (while $\rho$ goes down).
	In particular, we may find a pair of positions $j,j_\bullet$ (with $j<j_\bullet$) such that the subrun $\eta[j,j_\bullet]$ is augmenting
	and climbs up sufficiently high but not too high, and
	simultaneously its length $j_\bullet-j$ is relatively small.
	This is shown in Section~\ref{sec:go-up}.
\item
	While looking at $\rho$, maybe $\rho(j_\bullet)$ is in the middle of
	popping some $\beta$.
	For this reason, we move a bit forward, and we consider the next depth $\dplus(j_\bullet)$ visited by $\rho$ after $\rho(j_\bullet)$
	(formally, $\dplus(j_\bullet)$ is the smallest number $d$ such that $\rho$ visits $\Nabla{d}$ after $\rho(j_\bullet)$).
\item\label{pkt:6}
	We now want to exhibit a situation as in Lemma~\ref{lem:equiv-because-loop}.
	To this end, let us write $\eta(j)$ as $p_jX\mu$ (where $X\in\Gamma$ is a single stack symbol), and let $U=\eats{X}(p_j)$.
	\begin{enumerate}
	\item
		On the one side, recall that from $\rho(j_\bullet)$ there is a (short) run to $\Nabla{\dplus(j_\bullet)}$ (namely, a subrun of $\rho$),
		and from $\Nabla{\dplus(j_\bullet)}=r_{\dplus(j_\bullet)}\beta^{e-\dplus(j_\bullet)}\gamma$ there is an (again, short) run to $r\beta^{e-(\dplus(j_\bullet)+1)}\gamma$ for every $r\in V$
		(because, by assumption, $\eats{\beta}(r_{\dplus(j_\bullet)})=V$).
		Because $\rho(j_\bullet)\sim\eta(j_\bullet)$, for every $r\in V$ there is a corresponding (short) run from $\eta(j_\bullet)$ to a configuration bisimilar to $r\beta^{e-(\dplus(j_\bullet)+1)}\gamma$.
		Recall that $\eta(j_\bullet)$ is assumed to be appropriately higher than $\eta(j)$, so because the run is short, it cannot go down below $\eta(j)$ (and simultaneously it cannot also go up too much).
		In effect, for every $r\in V$ we have that $r\beta^{e-(\dplus(j_\bullet)+1)}\gamma\sim p_r\xi_r\mu$ for some control state $p_r$ and some (relatively small) stack content $\xi_r$.
	\item
		On the other side, by Lemma~\ref{lem:PDS-paths} for every $p\in U$ there is a very short run from $\eta(j)=p_jX\mu$ to $p\mu$.
		Because $\rho(j)\sim\eta(j)$, there is a corresponding (hence likewise short) run from $\rho(j)$ to a configuration bisimilar to $p\mu$.
		Recall that the stack content of $\rho(j)$ is higher than $\beta^{e-(\dplus(j_\bullet)+1)}\gamma$.
		It can be shown that the extremely short run from $\rho(j)$ to the configuration bisimilar to $p\mu$ cannot nivellate this difference,
		and thus we have $p\mu\sim r_p\chi_p\beta^{e-(\dplus(j_\bullet)+1)}\gamma$ for some control state $r_p$ and some stack content $\chi_p$
		(where the height of $\chi_p$ can be bounded appropriately).
	\end{enumerate}
	Thus, the assumptions of Lemma~\ref{lem:equiv-because-loop} are satisfied (with the exception that we only have one $\mu$ and one $\nu=\beta^{e-(\dplus(j_\bullet)+1)}$,
	while the lemma talks about $\mu_1,\mu_2,\nu_1,\nu_2$.
\item
	Recall that in Point~\ref{pkt:4} we have selected a short subrun of $\eta$ (described by $j$ and $j_\bullet$) that climbs up.
	In fact, the run $\eta$ is very long, and we can find not only one such subrun, but a lot of disjoint subruns having the above properties.
	Each of those subruns is described by some position $j$ (and by a corresponding position $j_\bullet$ coming soon after $j$).
	To every $j$ let us assign the tuple $(U,((r_p,\chi_p))_{p\in U},((p_r,\xi_r))_{r\in V})$, as defined in Point~\ref{pkt:6}.
	As already mentioned, the height of all the stack contents $\chi_p$ and $\xi_r$ can be bounded, so the tuple comes from a finite domain.
	If the length of $\eta$, hence the number of positions $j$, is large enough,
	by the pigeonhole principle we have two distinct positions $j_1,j_2$ (together with the corresponding positions $j_{1\bullet},j_{2\bullet}$) to which the same tuple was assigned.
	It causes that Lemma~\ref{lem:equiv-because-loop} can be applied; it implies that $r\beta^{e-(\dplus(j_{1\bullet})+1)}\gamma\sim r\beta^{e-(\dplus(j_{2\bullet})+1)}\gamma\sim r\beta^\omega$ for all $r\in V$.
	This contradicts the minimality of $e$, and thus finishes the proof.
\end{enumerate}

\section{Runs from $q\alpha\beta^h$ visit many classes}\label{sec:go-down}

In this section (namely, in Lemma~\ref{lem:go-down}) we prove that while ``going down'' from $q\alpha\beta^e\gamma$ it is possible to visit many bisimulation classes.

For analyzing control states along runs that pop a stack content
of the form $\alpha\beta^h$ we introduce a notion of a digging sequence.

\begin{definition}\label{def:digging-sequence}
	A {\em digging sequence}
	for $(q,\alpha,\beta)\in Q\times\Gamma^*\times\Gamma^*$
	is a sequence $(r_0,r_1,\dots,r_h)\in Q^{h+1}$ (for some $h\in\Nat$)
	that satisfies $r_0\in\eats{\alpha}(q)$,
	and $r_d\in\eats{\beta}(r_{d-1})$
	for all $d\in[1,h]$.
\end{definition}

Note that any such a digging sequence witnesses that
there is a run $\rho$ from $q\alpha\beta^h$
that visits $r_0\beta^h$,
then $r_1\beta^{h-1}$,
in fact every configuration $r_d\beta^{h-d}$ for $d\in[0,h]$,
thus ending in $r_h$.
The following definition captures the situation when
the intermediate configurations of a digging sequence are all
non-bisimilar to certain $\omega$-approximants
with stack content $\beta^\omega$.

\begin{definition}
	Let $q\alpha\beta^e\gamma$ (with $e\geq 1$) be a configuration of $\Pp$.
	We say that a digging sequence $(r_0,r_1,\dots,r_h)$ for $(q,\alpha,\beta)$ is \emph{$\beta^\omega$-avoiding}, if
	\begin{enumerate}
	\item	there exists $r'\in\eats{\beta}(r_0)$ such that $r'\beta^{e-1}\gamma\not\sim r'\beta^\omega$, and
	\item	$r_d\beta^{e-d}\gamma\not\sim r_d\beta^\omega$
		for all $d\in[1,h]$.
	\end{enumerate}
\end{definition}

Notice that we do not require $r_d\beta^{e-d}\gamma\not\sim r_d\beta^\omega$ for $d=0$.
If $h\geq 1$, Point 1 follows from Point 2 while taking $r'=r_1$
(but we need Point 1, if we want Lemma~\ref{lem:beta-omega-avoiding} to work for $h=0$).

The following lemma states that if
$q\alpha\beta^{e-d}\gamma\not\sim q\alpha\beta^\omega$ for all $d\in[1,e]$, then
there always exists
a $\beta^\omega$-avoiding digging sequence from $q\alpha\beta^e\gamma$, and any such a sequence of length $h$ that
is not maximal (i.e., $h<e$) can be prolonged to an $\beta^\omega$-avoiding digging
sequence of length $h+1$.

\begin{restatable}{lemma}{LemBetaOmegaAvoiding}\label{lem:beta-omega-avoiding}
	Let $(\alpha,\beta)$ be a linked pair, and let $q\alpha\beta^e\gamma$ (with $e\geq 1$) be a configuration of $\Pp$.
	If $q\alpha\beta^{e-d}\gamma\not\sim q\alpha\beta^\omega$ for all $d\in[1,e]$, then
	\begin{enumerate}
	\item	there exists a $\beta^\omega$-avoiding digging sequence $(r_0)$
		for $(q,\alpha,\beta)$, and
	\item	every $\beta^\omega$-avoiding digging sequence $(r_0,r_1,\dots,r_h)$
		for $(q,\alpha,\beta)$ with $h<e$
		can be extended by a control state $r_{h+1}$ to yield a longer $\beta^\omega$-avoiding digging sequence for $(q,\alpha,\beta)$.
	\end{enumerate}
\end{restatable}

\begin{proof}
	We start by proving Point 1.
	To this end, we should find a control state $r_0\in\eats{\alpha}(q)$ such that there exists $r'\in\eats{\beta}(r_0)$ satisfying $r'\beta^{e-1}\gamma\not\sim r'\beta^\omega$.
	Suppose that there is no such a control state $r_0$.
	This means that for all $r_0\in\eats{\alpha}(q)$ and for all $r'\in\eats{\beta}(r_0)$ we have $r'\beta^{e-1}\gamma\sim r'\beta^\omega$.
	However $\bigcup_{r_0\in\eats{\alpha}(q)}\eats{\beta}(r_0)=\eats{\alpha\beta}(q)=\eats{\alpha}(q)$,
	so $r'\beta^{e-1}\gamma\sim r'\beta^\omega$ for all $r'\in\eats{\alpha}(q)$.
	By Lemma~\ref{lem:equiv-because-equiv-below} this implies that $q\alpha\beta^{e-1}\gamma\sim q\alpha\beta^\omega$,
	contrarily to assumptions of the lemma.
	Thus, our supposition was false; there necessarily exists a control state $r_0$ as needed.

	We now come to Point 2.
	Take some $h<e$, and some $\beta^\omega$-avoiding digging sequence $(r_0,\dots,r_h)$.
	We should find a control state $r_{h+1}\in\eats{\beta}(r_h)$ such that $r_{h+1}\beta^{e-(h+1)}\gamma\not\sim r_{h+1}\beta^\omega$.
	If $h=0$, the existence of such a control state $r_1$ follows from Point 1 of Definition~\ref{def:digging-sequence}.
	In the case of $h\geq 1$, by Definition~\ref{def:digging-sequence} we have that $r_h\beta^{e-h}\gamma\not\sim r_h\beta^\omega$.
	But observe that if $r_{h+1}\beta^{e-(h+1)}\gamma\sim r_{h+1}\beta^\omega$ for all $r_{h+1}\in\eats{\beta}(r_h)$,
	then $r_h\beta^{e-h}\gamma\sim r_h\beta^\omega$ by Lemma~\ref{lem:equiv-because-equiv-below}, contradicting the above.
	This implies the existence of a control state $r_{h+1}$ as required.
\end{proof}

The following central lemma of this section states that if
$q\alpha\beta^{e-d}\gamma\not\sim q\alpha\beta^\omega$ for all $d\in[1,e]$, then there exists an $\omega$-avoiding digging
sequence from $q\alpha\beta^e\gamma$ in which every sufficiently long
subsequence contains
many different bisimulation classes.

\begin{lemma}\label{lem:go-down}
	There exists an elementary function $\iota\colon\Nat^2\rightarrow\Nat$
	such that for all $e\in\Nat$,
	all linked pairs $(\alpha,\beta)$, and all configurations $q\alpha\beta^e\gamma$ of $\Pp$ (with $e\geq 1$)
	such that
	$q\alpha\beta^{e-d}\gamma\not\sim q\alpha\beta^\omega$
	for all $d\in[1,e]$
	there is a $\beta^\omega$-avoiding digging sequence
	$(r_0,r_1,\dots, r_e)$ for $(q,\alpha,\beta)$
	such that the following holds:
	for every $k\in\Nat$ and every interval $[a,b]\subseteq[0,e]$ with $|[a,b]|>\iota(|\Pp|,k)$
	there are more than $k$ pairwise non-bisimilar configurations in
	$\set{r_d\beta^{e-d}\gamma\mid d\in\nobreak[a,b]}$.
\end{lemma}

\pp{The function $\iota$ is polynomial in $k$ and singly exponential in $|\Pp|$.}

\begin{proof}
	By Lemma~\ref{lem:beta-omega-avoiding} we know that
	there exists at least one
	$\beta^\omega$-avoiding digging sequence for $q\alpha\beta^e\gamma$.
	Among all such $\beta^\omega$-avoiding digging sequences
	we choose a particular one
	$(r_0,r_1,\dots,r_e)$ which we show to satisfy the statement of the lemma.
	We define the control states $r_0,\dots,r_e$ by induction.

	For the induction base, we choose an arbitrary $\beta^\omega$-avoiding digging sequence $(r_0)$;
	its existence is guaranteed by
	the first part of Lemma~\ref{lem:beta-omega-avoiding}.
	For the induction step, assume that we have already chosen
	$r_0,\ldots,r_{d-1}$ for some $d\in[1,e]$; we need
	to define $r_d$.
	Let $R_d=\set{r\in\eats{\beta}(r_{d-1})\mid r\beta^{e-d}\gamma\not\sim
	r\beta^\omega}$.
	We have $R_d\neq\emptyset$ by the second part of Lemma~\ref{lem:beta-omega-avoiding}.
	Among all control states in $R_d$ choose as $r_d$ some control state in
	$r\in R_d$ such that among the already determined sequence
	of configurations
	$r_0\beta^e\gamma,r_1\beta^{e-1}\gamma,\ldots r_{d-1}\beta^{e-(d-1)}\gamma$
	the class $\class{r\beta^{e-d}\gamma}$ does not occur, and if this is
	not possible, then this class appears last as early as possible;
	more formally, choose as $r_d$ any control state $r\in R_d$ that
	maximizes $\back(r,d)$, where
	\begin{align*}
		\back(r,d)=\min\set{j\in[1,d]\mid r_{d-j}\beta^{e-(d-j)}\gamma\sim r\beta^{e-d}}.
	\end{align*}
	Recall that $\min\emptyset=\omega$ by definition.
	This completes the construction.

	Let $V_0=\eats{\alpha}(q)$ and
	$V_d=\eats{\beta}(r_{d-1})$
	for all $d\in[1,e]$.
	We have $V_d\subseteq V_{d-1}$ for all $d\in[1,e]$
	by Lemma~\ref{lem:eats-beta}.
	Moreover $V_e\neq\emptyset$, because $r_e\in V_e$.

	We say that an interval $[a',b']\subseteq[0,e]$ is
	{\em monochromatic} if $V_d=V_{d'}$ for all $d,d'\in[a',b']$.
	By the above, $[0,e]$ (hence also every sub-interval of $[0,e]$) can be split into at most $|Q|$ monochromatic sub-intervals.

	To finish the proof, take some $k\in\Nat$ and some interval $[a,b]\subseteq[0,e]$ such that
	\begin{align}
		|[a,b]|>\iota(|\Pp|,k)=|\Pp|\cdot k\cdot((|\Pp|+k)^{|\Pp|}+1)\,.\label{eq:iota}
	\end{align}
	As said above, $[a,b]$ can be split into at most $|Q|$ monochromatic sub-intervals, so (recalling that $|Q|\leq|\Pp|$)
	there exists a monochromatic sub-interval $[a',b']\subseteq[a,b]$ of length
	$|[a',b']|>k\cdot(|Q|+k)^{|Q|}+1$.
	Below we prove that there are more than $k$ pairwise non-bisimilar configurations already in $\set{r_d\beta^{e-d}\gamma\mid d\in\nobreak[a',b']}$ .

	Let $V=V_{a'}=V_{b'}$ (i.e., $V=V_d$ for all $d\in[a',b']$).
	Let us first define the tuple of bisimulation classes
	$\theta_d=(\class{r\beta^{e-d}\gamma})_{r\in V}$
	for all $d\in[a',b']$.
	Let also $\Cc=\set{\class{r_d\beta^{e-d}\gamma}\mid d\in[a',b']}$
	be the set of bisimulation classes of
	configurations in $\set{r_d\beta^{e-d}\gamma\mid d\in[a',b']}$.
	Our goal is to prove that $|\Cc|>k$; for the sake of contradiction assume that $|\Cc|\leq k$.

	First, by the pigeonhole principle,
	since $|[a',b']|>k\cdot((|Q|+k)^{|Q|}+1)$
	and $|\Cc|\leq k$
	there exists one class
	$x\in\Cc$ for which there are
	$t+1= (|Q|+k)^{|Q|}+2$ indices $d_0,d_1,\dots,d_t$ such that
	$a'\leq d_0<d_1<d_2<\cdots <d_t\leq b'$ and
	$\class{r_{d_i}\beta^{e-d_i}\gamma}=x$
	for all $i\in[0,t]$.

	Secondly, let $\Oo=\set{\class{r\beta^\omega}\mid r\in V}$ be the set of
	bisimulation classes of the $\omega$-approximants.
	Observe that $|\Oo|\leq |Q|$.
	We claim that every component of every of the tuples
	\begin{align*}
		\theta_{d_1}=(\class{r\beta^{e-d_1}\gamma})_{r\in V},
		\dots,
		\theta_{d_t}=(\class{r\beta^{e-d_t}\gamma})_{r\in V}
	\end{align*}
	is a class inside $\Cc\cup\Oo$ (we do not claim this for
	$\theta_{d_0}$).
	Indeed, consider any index $d_i$ with $i\in[1,t]$
	and consider any $r\in V$.
	Recall that by the above construction, the control state
	$r_{d_i}$
	was defined to be one in
	\begin{align*}
		R_{d_i}=\set{r\in\eats{\beta}(r_{d_i-1})\mid r\beta^{e-d_i}\gamma\not\sim r\beta^\omega}
	\end{align*}
	such that $\back(r,d_i)$ is maximized.
	That is, for our control state
	$r\in V=\eats{\beta}(r_{d_i-1})$ we
	either have
	\begin{itemize}
	\item	$r\not\in R_{d_i}$, which implies that
		$r\beta^{e-d_i}\sim r\beta^\omega$ and thus
		$\class{r\beta^{e-d_i}\gamma}\in\Oo$, or
	\item $r\in R_{d_i}$ and
		$\back(r,d_i)\leq\back(r_{d_i},d_i)\leq d_i-d_{i-1}$, which
		implies that the class
		$\class{r\beta^{e-d_i}}$ can be found
		in
		$\set{\class{r_j\beta^{e-j}\gamma}\mid h\in[d_{i-1},d_i-1]}$ and thus in particular $\class{r\beta^{e-d_i}}\in\Cc$.
	\end{itemize}
	Thus, $\class{r\beta^{e-d_i}}\in\Cc\cup\Oo$ for all $i\in[1,t]$
	and all $r\in V$.

	Finally, because $t=(|Q|+k)^{|Q|}+1$ (and because there are at most $(|Q|+k)^{|Q|}$ tuples in $(\Cc\cup\Oo)^{|V|}$),
	by the pigeonhole principle there exist two
	distinct indices $s_1,s_2\in\set{d_1,\dots,d_t}$ such that
	$\theta_{s_1}=\theta_{s_2}$; say $s_1<s_2$.
	This means that $r\beta^{e-s_1}\gamma\sim r\beta^{e-s_2}\gamma$ for all $r\in V=V_{s_2}=\eats{\beta}(r_{s_2-1})$,
	which by Corollary~\ref{cor:repetition-is-omega} implies that $r_{s_2-1}\beta^{e-(s_2-1)}\sim r_{s_2-1}\beta^\omega$.
	Because $s_2>s_1\geq d_1>d_0\geq 0$, that is, $s_2-1\geq 1$,
	this implies that $r_{s_2-1}\not\in R_{s_2-1}$,
	contradicting our construction; thus $|\Cc|>k$.
\end{proof}

\section{Runs from $q\alpha\beta^\omega$ have to push}\label{sec:go-up}

In the previous section we were analysing runs starting from $q\alpha\beta^e\gamma$: 
we have shown that it is possible to pop the stack content and on the way visit many bisimulation classes.
In this section we prove that the $\omega$-approximant $q\alpha\beta^\omega$ is different: 
while popping the stack content from $q\alpha\beta^\omega$, one can visit only a small number of bisimulation classes.
Stating this conversely: every run starting from $q\alpha\beta^\omega$ that visit many bisimulation classes has to push a lot,
and actually it contains a subrun that is a concatenation of many nonempty augmenting runs, as stated by the following lemma.

\begin{restatable}{lemma}{LemGoUp}\label{lem:go-up}
	There exists an elementary function $g\colon\Nat^4\rightarrow\Nat$
	such that for every linked pair
	$(\alpha,\beta)$ and every
	run $\eta$ from $\alpha\beta^\omega$ visiting more than $g(|\Pp|,|\alpha|,|\beta|,k)$ distinct configurations,
	there exists a subrun
	$\eta_1\cdots\eta_k$ of $\eta$
	such that all
	$\eta_1,\ldots,\eta_k$ are augmenting and nonempty.
\end{restatable}

\sg{The function $g$ is single-exponentially bounded in its arguments, we could write this.}

\begin{proof}
	Take
	\begin{align}
		g(|\Pp|,|\alpha|,|\beta|,k)=(|\alpha|+|\beta|)\cdot|\Pp|^{k+2}\label{eq:def-g}
	\end{align}
	and consider a run $\eta$ from $\alpha\beta^\omega$ visiting more that $g(|\Pp|,|\alpha|,\allowbreak|\beta|,\allowbreak k)$ distinct configurations.

	Denote $N=|\eta|$.
	Let $-n$ be the minimal stack growth obtained by a prefix of $\eta$; formally
	\begin{align*}
		n=-\min\set{\sgrow(\eta[0,i])\mid i\in[0,N]}\,.
	\end{align*}
	Moreover, for $d\in[0,n]$, let $j_d$ denote the earliest positions in $\eta$ when the stack growth becomes $-d$; formally
	\begin{align*}
		j_d=\min\set{i\in[0,N]\mid\sgrow(\eta[0,i])\leq-d}\,.
	\end{align*}
	In particular, $j_0=0$.
	Additionally, let $j_{n+1}=N+1$.

	Clearly the stack growth cannot go down below $-d$ before becoming precisely equal to $-d$ at some earlier moment;
	thus $\sgrow(\eta[0,j_d])=-d$ for all $d\in[0,n]$.
	In effect, all the subruns $\eta[j_d,j_{d+1}-1]$ are augmenting, for $d\in[0,n]$
	(additionally, for $d<n$ we have $\sgrow(\eta[j_d,j_{d+1}-1])=0$, and $\sgrow(\eta[j_{d+1}-1,j_{d+1}])=-1$).
	Observe additionally that the stack content of $\eta[j_d]$ (for $d\in[0,n]$) is obtained from $\alpha\beta^\omega$ by popping some number of symbols.
	In effect, it is of the form $\kappa\beta^\omega$, where $\kappa\in\Gamma^+$ is either a suffix of $\alpha$ or a suffix of $\beta$.
	This means that among $\eta[j_d]$ there are at most $|Q|\cdot(|\alpha|+|\beta|)$ distinct configurations (we have to choose a control state in $Q$, and a nonempty suffix of $\alpha$ or $\beta$).
	Denote the set of these configurations by $\Cc$.

	Suppose first that $\sgrow(\eta[j_d,i])\leq k-1$ for all $d\in[0,n]$ and all $i\in[j_d,j_{d+1}-1]$.
	Then every configuration $\eta[i]$ visited by $\eta$ can be reached from a configuration $\eta[j_d]$ in $\Cc$
	by such a run $\eta[j_d,i]$, which is augmenting.
	In effect, Lemma~\ref{lem:radius} gives implies that there are at most $|\Cc|\cdot|\Pp|^{k-1+2}$ distinct configurations visited by $\eta$.
	But $|\Cc|\cdot|\Pp|^{k-1+2}\leq g(|\Pp|,|\alpha|,|\beta|,k)$, contrarily to our assumption.

	Thus, there exist $d\in[0,n]$ and $\ell\in[j_d,j_{d+1}-1]$ such that $\sgrow(\eta[j_d,\ell])\geq k$.
	Recall that $\eta[j_d,\ell]$ is augmenting.

	It is easy to split an augmenting run with stack growth at least $k$ into $k$ nonempty and augmenting subruns.
	Namely, for $e\in[0,k]$ we define
	\begin{align*}
		\ell_e=\max\set{i\in\eta[j_d,\ell]\mid \sgrow(\eta[j_d,i])\leq e}\,,
	\end{align*}
	and for $e\in[1,k]$ we take $\eta_e=\eta[\ell_{e-1},\ell_e]$.
	By definition $\eta_1\dots\eta_k$ is a subrun of $\eta$,
	and all $\eta_e$ are augmenting.
	Moreover, because $\Pp$ is in a push-pop normal form, all $\eta_e$ are nonempty:
	the stack growth cannot go above $e$ before becoming precisely equal to $e$ at some earlier moment.
\end{proof}

\section{The core of the proof\ppchanged{ of Theorem~\ref{thm:technical}}}\label{sec:core}

The goal of this section is finish proving our main technical theorem, Theorem~\ref{thm:technical}.

Recall that Theorem~\ref{thm:technical} claims the existence of a function $h$.
For reasons of readability we postpone the definition of this function.

Towards a proof of Theorem~\ref{thm:technical} let us
fix a control state $q\in Q$, a linked pair
$(\alpha,\beta)$, and a stack content $\gamma\in\Gamma^*$.
In order to simplify some formulae, assume moreover that $|\beta|\geq\max\set{|\alpha|,\Ee}$,
where the constant $\Ee\leq\exp(|\Pp|)$ is taken from Lemma~\ref{lem:PDS-paths},
and is fixed for the rest of this section.
The case of $|\beta|<\max\set{|\alpha|,\Ee}$ can be easily reduced to the former one by considering $\beta'=\beta^{\max\set{|\alpha|,\Ee}}$, as shown below.

\begin{lemma}\label{lem:wlog}
	If Theorem~\ref{thm:technical} holds when $|\beta|\geq\max\set{|\alpha|,\Ee}$, then it holds in general.
\end{lemma}

\begin{proof}
	Suppose that we have already defined a function $h(|\Pp|,\allowbreak|\alpha|,\allowbreak|\beta|)$ for arguments such that $|\beta|\geq\max\set{|\alpha|,\Ee}$,
	so that Theorem~\ref{thm:technical} holds in this case.
	For the remaining arguments we take $h(|\Pp|,|\alpha|,|\beta|)=h(|\Pp|,|\alpha|,|\beta|\cdot\max\set{|\alpha|,\Ee})\cdot \max\set{|\alpha|,\Ee}$.

	Take now a control state $q\in Q$, a linked pair $(\alpha,\beta)$, and a stack content $\gamma\in\Gamma^*$ such that $|\beta|<\max\set{|\alpha|,\Ee}$.
	Let $\beta_\diamond=\beta^{\max\set{|\alpha|,\Ee}}$.
	Observe that $(\alpha,\beta_\diamond)$ is a linked pair as well.
	Clearly
	\begin{align*}
		\set{q\alpha\beta_\diamond^i\gamma\mid i\in\Nat}&=\set{q\alpha\beta^{i\cdot\max\set{|\alpha|,\Ee}}\gamma\mid i\in\Nat}\\
			&\subseteq\set{q\alpha\beta^i\gamma\mid i\in\Nat}\,,
	\end{align*}
	so if the letter set contains only finitely many pairwise non-bisimilar configurations, so does the former.
	Moreover, $|\beta_\diamond|\geq\max\set{|\alpha|,\Ee})$, so we can use Theorem~\ref{thm:technical} for $(\alpha,\beta_\diamond)$,
	obtaining that
	\begin{align*}
		q\alpha\beta^{e\cdot\max\set{|\alpha|,\Ee}}\gamma=q\alpha\beta_\diamond^e\gamma\sim q\beta_\diamond^\omega=q\beta^\omega
	\end{align*}
	for some $e\leq h(|\Pp|,|\alpha|,|\beta'|)$.
	This gives the thesis, since from our definition of $h$ it immediately follows that 
	\begin{align*}
		e\cdot\max\set{|\alpha|,\Ee}\leq h(|\Pp|,|\alpha|,|\beta|)\,.\tag*{\qedhere}
	\end{align*}
\end{proof}

Let $e$ be the smallest natural number satisfying
$q\alpha\beta^e\gamma\sim q\beta^\omega$.
The goal is to prove that $e\leq h(|\Pp|,|\alpha|,|\beta|)$.
For the sake of contradiction, assume $e$ is sufficiently large, namely
$e>h(|\Pp|,|\alpha|,|\beta|)$, where, as mentioned above, the function $h$ is
defined later.
By Lemma~\ref{lem:go-down} there exists a $\beta^\omega$-avoiding
digging sequence
$\sigma=(r_0,r_1,\dots,r_e)$ for $(q,\alpha,\beta)$,
visiting many pairwise non-bisimilar configurations (as specified in Lemma~\ref{lem:go-down}).
By Definition~\ref{def:digging-sequence}, $r_0\in\eats{\alpha}(q)$ and $r_d\in\eats{\beta}(r_{d-1})$
for all $d\in[1,e]$.

\paragraph{The run $\rho$.}

Let us fix a run $\rho$ from
$q\alpha\beta^e\gamma$
to $r_e\gamma$ that
realizes $\sigma$ as quickly as
possible in the following sense:
there exists a factorization
$\rho=\rho_0\rho_1\cdots\rho_e$
such that
\begin{itemize}
\item $q\alpha\xrightarrow{\rho_0}r_0$ is a shortest
	run from $q\alpha$ to
	$r_0$ and
\item $r_{d-1}\beta\xrightarrow{\rho_d}r_d$ is a shortest
	run from $r_{d-1}\beta$ to $r_d$
	for all $d\in[1,e]$.
\end{itemize}
Note that $|\rho_0|\leq|\alpha|\cdot\Ee$
and $1\leq|\beta|\leq|\rho_d|\leq|\beta|\cdot\Ee$ for all $d\in[1,e]$, by Lemma~\ref{lem:PDS-paths}.
Because by assumption $|\alpha|\leq|\beta|$, we actually have $|\rho_d|\leq|\beta|\cdot\Ee$ for all $d\in[0,e]$.

Let $N$ denote the length of the run $\rho$, $N=|\rho|$.
Since $e$ was chosen sufficiently large for the arguments
in this section to work, so is $N$.
Let us also fix the following intermediate configurations
\begin{align*}
	\Nabla{d}=r_d\beta^{e-d}\gamma\text{ for all $d\in[0,e]$}
\end{align*}
that the run $\rho$ visits.
We call these configurations {\em $\nabla$-con\-fig\-u\-ra\-tions}.
For every position $i\in[0,N-1]$ let
\begin{align*}
	\dplus(i)=\min\set{d\in [0,e]\mid
		\exists j\in[i+1,N].\,\rho(j)=\Nabla{d}}
\end{align*}
denote the ``depth index $d$'' of the next $\nabla$-configuration $\Nabla{d}$
the run $\rho$ sees strictly after position $i$.

Next we have two observations.

\begin{restatable}{observation}{ObsOdb}\label{O d b-a}
	For $0\leq i\leq i'\leq N-1$, we have that
	$\left\lfloor\frac{i'-i}{|\beta|\cdot\Ee}\right\rfloor\leq \dplus(i')-\dplus(i)\leq i'-i$.
\end{restatable}

\begin{proof}
	Immediate consequence of the definition of $\dplus(\cdot)$ and the inequalities $1\leq|\rho_d|$ for $d\in[1,e]$ and $|\rho_d|\leq|\beta|\cdot\Ee$ for $d\in[0,e]$.
\end{proof}

\begin{restatable}{observation}{ObsOnablapath}\label{O nabla path}
	For all $i\in[0,N-1]$, we have
	\begin{enumerate}
		\item $\rho(i)=q_i\kappa_i\beta^{e-\dplus(i)}\gamma$ for
			some control state
			$q_i\in Q$ and some stack content $\kappa_i\in\Gamma^*$
			such that $1\leq|\kappa_i|\leq|\beta|\cdot\Ee$, and
		\item if $\dplus(i)+1\leq e$, then
			$\dist(\rho(i),r\beta^{e-(\dplus(i)+1)}\gamma)\leq
			2\cdot|\beta|\cdot\Ee$ for all $r\in V_{\dplus(i)+1}$.
	\end{enumerate}
\end{restatable}

\begin{proof}
	For Point 1 observe that, by construction, $\Nabla{\dplus(i)}=r_{\dplus(i)}\beta^{e-\dplus(i)}\gamma$ is the earliest configuration of
	the run $\rho$ whose stack content is of the form $\beta^{e-\dplus(i)}\gamma$;
	the upper bound on $|\kappa_i|$ follows from the inequality
	$|\rho_{\dplus(i)}|\leq|\beta|\cdot\Ee$, and from the fact that $\rho_{\dplus(i)}$ can pop at most one symbol in every step.
	Point 2 is a consequence of Point 1 and Lemma~\ref{lem:PDS-paths}, if we recall that $\Nabla{\dplus(i)}\to^*r\beta^{e-(\dplus(i)+1)}\gamma$ for all $r\in V_{\dplus(i)+1}$.
\end{proof}

\paragraph{The run $\eta$ and its analysis.}

Since $q\alpha\beta^e\gamma\sim
q\alpha\beta^\omega$ there must exist
a run $\eta$ from $q\alpha\beta^\omega$
that mimics the run $\rho$, that is,
such that $\rho(i)\sim\eta(i)$ for all
$i\in[0,N]$.

Similarly as in the proof of Lemma \ref{lem:go-down}, define
$V_0=\eats{\alpha}(q)$ and
$V_d=\eats{\beta}(r_{d-1})$ for all $d\in[1,e]$.
Recall that an interval $[a,b]\subseteq[0,e]$ is called monochromatic if $V_d=V_{d'}$ for all $d,d'\in[a,b]$.

Let us define the constant
\begin{align}
	\ell=\ell(|\Pp|,|\beta|,\Ee)=
	\Big(2+\iota\big(|\Pp|,|\Pp|^{2\cdot|\beta|\cdot\Ee+1}\big)\Big)\cdot|\beta|\cdot\Ee\,,\label{def ell}
\end{align}
where $\iota$ is the function from Lemma~\ref{lem:go-down}.
Anticipating that the difference $N-\ell$ is
sufficiently large,
let $\Jj$ be the set of positions $j\in[0,N-\ell-1]$ such that $\eta[j,j+\ell]$ is augmenting, and $[\dplus(j),\dplus(j+\ell)]$ is monochromatic.

\sg{The constant $\ell$ is triply-exponentially bounded.}

For all positions $j\in\Jj$ let $j_\bullet$ denote the
earliest position inside $\rho$ after $j$ when the augmenting
run $\eta[j,j+\ell]$ has pushed $2\cdot|\beta|\cdot\Ee$ stack symbols; formally
\begin{align*}
	j_\bullet=\min\set{i\in[j,j+\ell]\mid&\\
		&\hspace{-2em}\sgrow(\eta[j,i])=2\cdot|\beta|\cdot\Ee}.
\end{align*}
That is, we can write the configurations $\eta(j)$ and $\eta(j_\bullet)$ as
\begin{align}
	&\eta(j)=p_j X_j\mu_j\,,
		&&\mbox{where $p_j\in Q$, and $X_j\in\Gamma$,}\label{eta j}\\
		&&&\mbox{and $\mu_j\in\RegStr(\Gamma)$, and}\nonumber\\
	&\eta(j_\bullet)=p_{j_\bullet}\zeta_{j_\bullet}\mu_j\,,
		&&\mbox{where $p_{j_\bullet}\in Q$, and $\zeta_{j_\bullet}\in\Gamma^*$}\label{eta ol j}\\
		&&&\mbox{is such that $|\zeta_{j_\bullet}|=2\cdot|\beta|\cdot\Ee+1$.}\nonumber
\end{align}

The conclusion of Lemma~\ref{lem:go-down} together with Lemma~\ref{lem:radius} imply that $j_\bullet$ is well-defined (i.e., $j_\bullet\in\Nat$), and that actually
$j_\bullet\leq j+\ell-|\beta|\cdot\Ee$, as shown in the following lemma.

\begin{lemma}\label{lemma:j}
	If $j\in\Jj$, then $j_\bullet\leq j+\ell-|\beta|\cdot\Ee$.
\end{lemma}

\begin{proof}
	Denote
	\begin{align}
		\ell'=\ell-|\beta|\cdot\Ee\,.\label{def:ell'}
	\end{align}
	The set of configurations
	$\set{\rho(i)\mid i\in [j,j+\ell']}$
	contains a set of configurations
	$\set{\Nabla{d}\mid d\in[\dplus(j),\dplus(j+\ell')-1]}$.
	Moreover,
	\begin{eqnarray*}
		|[\dplus(j),\dplus(j+\ell')-1]|\hspace{-8.8em}&&\\
			&\hspace{-0.5em}=\hspace{-0.5em}&\dplus(j+\ell')-\dplus(j)\\
			&\hspace{-0.5em}\stackrel{\text{Observation~\ref{O d b-a}}}{\geq}\hspace{-0.5em}&\left\lfloor\frac{j+\ell'-j}{|\beta|\cdot\Ee}\right\rfloor\\
			&\hspace{-0.5em}\stackrel{\text{\eqref{def:ell'}, \eqref{def ell}}}{=}\hspace{-0.5em}
				&\left\lfloor\frac{\big(2+\iota\big(|\Pp|,|\Pp|^{2\cdot|\beta|\cdot\Ee+1}\big)\big)\cdot|\beta|\cdot\Ee-|\beta|\cdot\Ee}{|\beta|\cdot\Ee}\right\rfloor\\
			&\hspace{-0.5em}=\hspace{-0.5em}&1+\iota\big(|\Pp|,|\Pp|^{2\cdot|\beta|\cdot\Ee+1}\big)\,.
	\end{eqnarray*}
	In consequence, by the conclusion of Lemma~\ref{lem:go-down} there are more than $|\Pp|^{2\cdot|\beta|\cdot\Ee+1}$ pairwise non-bisimilar configurations in
	$\set{\Nabla{d}\mid d\in[\dplus(j),\dplus(j+\ell')-1]}$, thus in particular in $\set{\rho(i)\mid i\in [j,j+\ell']}$.

	Recall now that $\rho(i)\sim\eta(i)$, hence also in $\set{\eta(i)\mid i\in [j,j+\ell']}$ there are more than
	$|\Pp|^{2\cdot|\beta|\cdot\Ee+1}$ pairwise non-bisimilar configurations;
	in particular,
	\begin{align}
		|\set{\eta(i)\mid i\in [j,j+\ell']}|>|\Pp|^{2\cdot|\beta|\cdot\Ee+1}\,.\label{eq:many-conf}
	\end{align}
	By assumption $\eta[j,j+\ell]$ is augmenting, hence for every $i\in [j,j+\ell']$ its subrun $\eta[j,i]$ is an augmenting run from $\eta(j)$ to $\eta(i)$.
	By Lemma~\ref{lem:radius} (used with $z=2\cdot|\beta|\cdot\Ee-1$), the number of configurations $\eta(i)$ with $i\in [j,j+\ell']$ such that $\sgrow(\eta[j,i])\leq 2\cdot|\beta|\cdot\Ee-1$ is at most
	$|\Pp|^{2\cdot|\beta|\cdot\Ee+1}$.
	Thus, due to Inequality~\eqref{eq:many-conf}, there is a position $i\in [j,j+\ell']$ such that $\sgrow(\eta[j,i])\geq 2\cdot|\beta|\cdot\Ee$.
	Because $\Pp$ is in a push-pop normal form, the stack height changes only by one in each step,
	thus necessarily there is also a position $i\in [j,j+\ell']$ such that $\sgrow(\eta[j,i])$ equals $2\cdot|\beta|\cdot\Ee$;
	the smallest such $i$ is taken as $j_\bullet$, hence $j_\bullet\leq\ell'$.
\end{proof}

As a consequence, we obtain that
\begin{align}
	\dplus(j_\bullet)+1\leq \dplus(j+\ell)\leq \dplus(j)+\ell\,.\label{eq:j-bullet}
\end{align}
for all $j\in\Jj$. Indeed,
\begin{eqnarray*}
	\dplus(j+\ell)-\dplus(j_\bullet)
		&\hspace{-0.5em}\stackrel{\text{Observation~\ref{O d b-a}}}{\geq}\hspace{-0.5em}&
		\left\lfloor\frac{(j+\ell)-j_\bullet}{|\beta|\cdot\Ee}\right\rfloor\\
		&\hspace{-0.5em}\stackrel{\text{Lemma~\ref{lemma:j}}}{\geq}\hspace{-0.5em} &
		\left\lfloor\frac{(j+\ell)-(j+l-|\beta|\cdot\Ee)}{|\beta|\cdot\Ee}\right\rfloor\\
		&\hspace{-0.5em}=\hspace{-0.5em}&1\,,
\end{eqnarray*}
establishing the first inequality.
The second inequality is a direct consequence of Observation~\ref{O d b-a}:
\begin{align*}
	\dplus(j+\ell)-\dplus(j)\leq (j+\ell)-j=\ell\,.\tag*{\qedhere}
\end{align*}

Inequalities~\eqref{eq:j-bullet} imply that
$\dplus(j_\bullet)+1$ belongs to the mo\-no\-chro\-ma\-tic segment $[\dplus(j),\dplus(j+\ell)]$, that is, $V_{\dplus(j_\bullet)+1}=V_{\dplus(j)}$.

\paragraph{Inter-dependencies of bisimulation classes near $\eta(j)$ and $\rho(j)$.}

Fix any $j\in\Jj$.
Our goal is to establish mutual dependencies between configurations near $\eta(j)$ and $\rho(j)$, allowing us to use Lemma~\ref{lem:equiv-because-loop}.
Firstly, by Observation~\ref{O nabla path},
$\dist(\rho(j_\bullet),r\beta^{e-(\dplus(j_\bullet)+1)}\gamma)\leq
2\cdot|\beta|\cdot\Ee$
for all $r\in V_{\dplus(j_\bullet)+1}=V_{\dplus(j)}$.
Because $\rho(j_\bullet)\sim\eta(j_\bullet)$, there is a configuration bisimilar to $r\beta^{e-(\dplus(j_\bullet)+1)}\gamma$ in distance at most $2\cdot|\beta|\cdot\Ee$ from $\eta(j_\bullet)$.
Recall from Equation~\eqref{eta ol j} that $\eta(j_\bullet)=p_{j_\bullet}\zeta_{j_\bullet}\mu_j$,
where $|\zeta_{j_\bullet}|=2\cdot|\beta|\cdot\Ee+1$.
Hence, denoting $U_j=\eats{X_j}(p_j)\subseteq Q$, for all $r\in V_{\dplus(j)}$ we have that
\begin{align}
	r\beta^{e-(\dplus(j_\bullet)+1)}\gamma\sim p_{r,j}\xi_{r,j}\mu_j\label{p rV}
\end{align}
for some $p_{r,j}\in Q, \xi_{r,j}\in\Gamma^+$ with
\begin{align}
	1\leq|\xi_{r,j}|\leq 4\cdot|\beta|\cdot\Ee+1\label{size xi prV}
\end{align}
and
\begin{align}
	\eats{\xi_{r,j}}(p_{r,j})\subseteq\eats{X_j}(p_j)=U_j\,.\label{in Uj}
\end{align}
Inclusion~\eqref{in Uj} follows from Lemma~\ref{lem:extend-eats}: the subrun $\eta[j,j_\bullet]$ composed with a run witnessing $\dist(\eta(j_\bullet),p_{r,j}\xi_{r,j}\mu_j)\leq2\cdot|\beta|\cdot\Ee$
can be seen as a run from $p_jX_j$ to $p_{r,j}\xi_{r,j}$.

Secondly, recall from Equation~\eqref{eta j} that $\eta(j)=p_jX_j\mu_j$.
By Lemma~\ref{lem:PDS-paths} we have
$\dist(\eta(j),p\mu_j)\leq\Ee$ for all $p\in U_j$.
Because $\rho(j)\sim\eta(j)$, there is a configuration bisimilar to $p\mu_j$ in distance
at most $\Ee$ from $\rho(j)$, for all $p\in U_j$.
Observation~\ref{O nabla path} states
that the stack content of $\rho(j)$ is of the form
$\kappa_j\beta^{e-\dplus(j)}\gamma$ with
$1\leq|\kappa_j|\leq|\beta|\cdot\Ee$, that is,
of the form $\kappa_j'\beta^{e-(\dplus(j_\bullet)+1)}\gamma$ with
$1+|\beta|\cdot(\dplus(j_\bullet)+1-\dplus(j))\leq|\kappa_j'|\leq|\beta|\cdot(\dplus(j_\bullet)+1-\dplus(j)+\Ee)$.
Recall that, by assumption $|\beta|\geq\Ee$ (and that, by definition, $\dplus(j_\bullet)\geq \dplus(j)$),
thus in particular $|\kappa_j'|\geq1+\Ee$.
On the other hand $\dplus(j_\bullet)+1-\dplus(j)\leq\ell$ by Inequality~\eqref{eq:j-bullet}, so $1+\Ee\leq|\kappa_j'|\leq|\beta|\cdot(\ell+\Ee)$.
Thus, for all $p\in U_j$,
\begin{align}
	p\mu_j\sim r_{p,j}\chi_{p,j}\beta^{e-({\dplus(j_\bullet)+1})}\gamma\label{p muj}
\end{align}
for some $r_{p,j}\in Q,\chi_{p,j}\in\Gamma^*$ with
\begin{align}
	1\leq|\chi_{p,j}|\leq|\beta|\cdot(\ell+\Ee)+\Ee\label{size chi pUj}
\end{align}
and
\begin{align}
	\eats{\chi_{p,j}}(r_{p,j})\subseteq V_{\dplus(j)}\,.\label{in V}
\end{align}
Inclusion~\eqref{in V} follows from Lemma~\ref{lem:extend-eats},
where we use the subrun $\rho_{\dplus(j)}$ composed with a run witnessing $\dist(\rho(j),\allowbreak r_{p,j}\chi_{p,j}\beta^{e-({\dplus(j_\bullet)+1})}\gamma)\leq\Ee$.
If $\dplus(j)\geq 1$, this composition can be seen as a run from $r_{\dplus(j)-1}\beta^{(\dplus(j_\bullet)+1)-(\dplus(j)-1)}$ to $r_{p,j}\chi_{p,j}$,
and $\eats{\beta^{(\dplus(j_\bullet)+1)-(\dplus(j)-1)}}(r_{\dplus(j)-1})=\eats{\beta}(r_{\dplus(j)-1})=V_{\dplus(j)}$;
otherwise (i.e., if $\dplus(j)=0$), it can be seen as a run from $q\alpha\beta^{\dplus(j_\bullet)+1}$ to $r_{p,j}\chi_{p,j}$,
and $\eats{\alpha\beta^{\dplus(j_\bullet)+1}}(q)=\eats{\alpha}(q)=V_0=V_{\dplus(j)}$.

\paragraph{Inter-dependencies cannot happen too often.}

For every $j\in\Jj$ let us define
\begin{align*}
	\col(j)=(U_j,((r_{p,j},\chi_{p,j}))_{p\in U_j},V_{\dplus(j)},((p_{r,j},\xi_{r,j}))_{r\in V_{\dplus(j)}})\,.
\end{align*}
Inequalities~\eqref{size xi prV} and~\eqref{size chi pUj} 
allow us to bound the number of all possible ``colors''.

\begin{lemma}\label{lem:h1}
	There are at most
	\begin{align*}
		h_1=h_1(|\Pp|,|\beta|,\Ee,\ell)=|\Pp|^{(1+|\beta|\cdot(\ell+\Ee)+\Ee+1+4\cdot|\beta|\cdot\Ee+1)\cdot|\Pp|}
	\end{align*}
	many different values for $\col(j)$, for $j\in\Jj$.
\end{lemma}

\begin{proof}
	By Inequalities~\eqref{size xi prV} and~\eqref{size chi pUj} the stack contents $\chi_{p,j}$ and $\xi_{r,j}$ appearing in the definition of $\col(j)$ satisfy
	\begin{align*}
	 	1\leq|\chi_{p,j}|\leq|\beta|\cdot(\ell+\Ee)+\Ee
	\end{align*}
	and
	\begin{align*}
		1\leq|\xi_{r,j}|\leq 4\cdot|\beta|\cdot\Ee+1\,.
	\end{align*}

	Thus, for every control state $p\in Q$ we either choose that $p\not\in U_j$, or we choose
	a control state $r_{p,j}\in Q$ and a stack content $\chi_{p,j}\in\Gamma^+$ of length at most $|\beta|\cdot(\ell+\Ee)+\Ee$;
	this gives us
	\begin{align*}
		1+|Q|\cdot\big((|\Gamma|+1)^{|\beta|\cdot(\ell+\Ee)+\Ee}-1\big)\leq |\Pp|^{1+|\beta|\cdot(\ell+\Ee)+\Ee}
	\end{align*}
	possibilities (tacitly assuming that
	$|\Pp|\geq1$).

	Likewise, for every control state $r\in Q$ we either choose that $r\not\in V_{\dplus(j)}$, or we choose
	a control state $p_{r,j}\in Q$ and a stack content $\xi_{r,j}\in\Gamma^+$ of length at most $4\cdot|\beta|\cdot\Ee+1$;
	this gives us
	\begin{align*}
		1+|Q|\cdot\big((|\Gamma|+1)^{4\cdot|\beta|\cdot\Ee+1}-1\big)\leq |\Pp|^{1+4\cdot|\beta|\cdot\Ee+1}
	\end{align*}
	possibilities.

	Because we are going to do both choices for every control state in $Q$, we multiply the two numbers (i.e., we add their exponents),
	and we take their $|Q|$-th power (i.e., we multiply the exponent by $|Q|$).
	Taking into account that $|Q|\leq|\Pp|$, this gives us the formula from the lemma statement.
\end{proof}

\sg{$h_1$ is four-fold exponentially bounded.}

\sg{$h$ is five-fold exponentially bounded.}

By Lemmata~\ref{lem:go-down} and~\ref{lem:go-up}, and by a direct calculation we obtain the following lemma.

\begin{lemma}\label{lem:e-big}
	If $e>h(|\Pp|,|\alpha|,|\beta|)$, where
	\begin{align*}
		h(|\Pp|,|\alpha|,|\beta|)=\iota(|\Pp|,g(|\Pp|,|\alpha|,|\beta|,(h_1+|\Pp|)\cdot\ell+1))\,,
	\end{align*}
	then there exist two positions $j_1,j_2\in\Jj$ with
	$j_2\geq j_1+\ell$ such that $\col(j_1)=\col(j_2)$.
\end{lemma}

\begin{proof}
	Firstly, by the conclusion of Lemma~\ref{lem:go-down} (used for the interval $[0,e]$),
	there are more than $g(|\Pp|,|\alpha|,|\beta|,(h_1+|\Pp|)\cdot\ell+1)$ pairwise non-bisimilar configurations in
	\begin{align*}
		\set{r_d\beta^{e-d}\gamma\mid d\in[0,e]}\subseteq\set{\rho(i)\mid i\in[0,N]}\,.
	\end{align*}
	Because $\rho(i)\sim\eta(i)$ for all $i\in[0,N]$, the same number of pairwise non-bisimilar configurations can be found in $\set{\eta(i)\mid i\in[0,N]}$;
	in particular, $\eta$ visits more than $g(|\Pp|,|\alpha|,|\beta|,(h_1+|\Pp|)\cdot\ell+1)$ distinct configurations.

	Secondly, take $k=(h_1+|\Pp|)\cdot\ell+1$.
	By Lemma~\ref{lem:go-up}, a subrun of $\eta$ can be represented as $\eta_1\dots\eta_k$, where all $\eta_1,\dots,\eta_k$ are augmenting and nonempty.
	Let $\Jj_1$ be the set of positions of $\eta$ at which the augmenting subruns $\eta_i$ start;
	formally, let $\Jj_1=\set{j_0+|\eta_1\dots\eta_{i-1}|:i\in[1,k]}$, where $j_0$ is such that $\eta[j_0,j_0+|\eta_1\dots\eta_k|]=\eta_1\dots\eta_k$.
	Observe that $\eta[j,\max \Jj_1]$ (as a composition of augmenting runs) is augmenting for every $j\in\Jj_1$, hence also every its prefix is augmenting.
	Moreover, $|\Jj_1|=k$ and $\max\Jj_1\leq N-1$.

	Thirdly, let $\Jj_2$ contain the $1$-st, the $(\ell+1)$-th, the $(2\ell+1)$-th, ..., the $((h_1+|\Pp|-1)\cdot\ell+1)$-th element of $\Jj_1$.
	Notice that $|\Jj_2|=h_1+|\Pp|$, and that $j+\ell\leq\max\Jj_1$ for all $j\in\Jj_2$, which implies that $\eta[j,j+\ell]$ is augmenting.
	Moreover, $j'\geq j+\ell$ whenever $j,j'\in\Jj_2$ and $j<j'$.

	By monotonicity of $\dplus(\cdot)$, the last property implies that $\dplus(j')\geq \dplus(j+\ell)$ whenever $j,j'\in\Jj_2$ and $j<j'$.
	In other words, the intervals $[\dplus(j),\dplus(j+\ell)]$ are almost disjoint for different positions $j\in\Jj_2$: only the last element of one interval can be the first element of another interval.
	In effect, for at most $|Q|-1$ positions $j\in\Jj_2$, the interval $[\dplus(j),\dplus(j+\ell)]$ can be non-monochromatic (recall that the whole $[0,e]$ can be split into at most $|Q|$ monochromatic intervals).
	For the remaining $h_1+|\Pp|-(|Q|-1)\geq h_1+1$ positions $j\in\Jj_2$, the interval $[\dplus(j),\dplus(j+\ell)]$ is monochromatic, and thus $j\in\Jj$
	(recall that, by definition, $\Jj$ contains those positions $j\in[0,N-\ell-1]$ for which $\eta[j,j+\ell]$ is augmenting and $[\dplus(j),\dplus(j+\ell)]$ is monochromatic).

	We thus have more than $h_1$ positions in $\Jj\cap\Jj_2$.
	Because by Lemma~\ref{lem:h1} there at most $h_1$ many different values for $\col(\cdot)$,
	by the pigeonhole principle there are necessarily two distinct positions $j_1,j_2\in\Jj\cap\Jj_2$ such that $\col(j_1)=\col(j_2)$.
	This finishes the proof, since $j_1<j_2$ for elements of $\Jj_2$ implies $j_2\geq j_1+\ell$.
\end{proof}

Let us thus fix
two positions $j_1,j_2\in\Jj$ with
$j_2\geq j_1+\ell$ such that
\begin{align*}
	\col(j_1)=\col(j_2)=(U,((r_p,\chi_p))_{p\in U},V,((p_r,\xi_r))_{r\in V}).
\end{align*}
By Condition~\eqref{p muj} we obtain
\begin{align*}
	p\mu_{j_i}\sim r_p\chi_p\beta^{e-({\dplus(j_{i\bullet})+1)}}\gamma
	\quad\text{for all } p\in U\text{ and all $i\in\set{1,2}$,}
\end{align*}
and by Condition~\eqref{p rV} we obtain
\begin{align*}
	r\beta^{e-(\dplus(j_{{i}\bullet})+1)}\sim p_r\xi_r\mu_{j_i}
	\qquad\text{for all } r\in V\text{ and all $i\in\set{1,2}$.}
\end{align*}
Since moreover
$\eats{\xi_r}(p_r)\subseteq U$ by Inclusion~\eqref{in Uj} and
$\eats{\chi_p}(r_p)\subseteq V$ by Inclusion~\eqref{in V}, we can apply
Lemma~\ref{lem:equiv-because-loop} and obtain (by setting $\mu_i=\mu_{j_i}$
and $\nu_i=\beta^{e-(\dplus(j_{i\bullet})+1)}\gamma$) that
\begin{align*}
	r\beta^{e-(\dplus(j_{1\bullet})+1)}\gamma\sim r\beta^{e-(\dplus(j_{2\bullet})+1)}\gamma\qquad{\text{for all $r\in V$}}.
\end{align*}
Recall that $V=V_{\dplus(j_{2\bullet})+1}=\eats{\beta}(r_{\dplus(j_{2\bullet})})$.
Since $\dplus(j_{1\bullet})+1\leq \dplus(j_1+l)\leq \dplus(j_{2\bullet})$ (i.e., $\dplus(j_{1\bullet})\neq \dplus(j_{2\bullet})$) by Inequality~\eqref{eq:j-bullet}, by Corollary~\ref{cor:repetition-is-omega} we obtain that
$r_{\dplus(j_{2\bullet})}\beta^{e-\dplus(j_{2\bullet})}\gamma\sim r_{\dplus(j_{2\bullet})}\beta^\omega$,
contrarily to the assumption that our digging sequence $\sigma$ is $\beta^\omega$-avoiding.
It follows that necessarily $e\leq h(|\Pp|,|\alpha|,|\beta|)$, finishing the proof.

\section{Conclusion}\label{sec:conclusion}

We have shown that any bisimulation-finite PDS is
already bisimilar to a finite system of size elementary in the size of the PDS.
A careful analysis reveals that the function $\varphi$ in
Theorem~\ref{thm:main} is in fact six-fold
exponential in the size of the PDS.

Indeed, recall that the constant $\Ff$ from Lemma~\ref{lem:decompose} is doubly exponential in $|\Pp|$;
in consequence $|\alpha|$ and $|\beta|$ for all considered linked pairs are at most doubly exponential in $|\Pp|$.
The constant $\Ee$ from Lemma~\ref{lem:PDS-paths} is singly exponential in $|\Pp|$.
The function $\iota(|\Pp|,k)$ from Lemma~\ref{lem:go-down} is polynomial in the second argument, and singly exponential in $|\Pp|$ (consult Equality~\eqref{eq:iota}).
In consequence, the constant $\ell$ defined in Equality~\eqref{def ell} is at most triply exponential in $|\Pp|$,
and thus the constant $h_1$ defined in Lemma~\ref{lem:h1} is at most four-fold exponential in $|\Pp|$.
The function $g$ defined in Equality~\eqref{eq:def-g} is singly exponential in its arguments,
and thus the value $h(|\Pp|,|\alpha|,|\beta|)$ defined in Lemma~\ref{lem:e-big} is at most five-fold exponential in $|\Pp|$.
This implies that the value $\lambda(|\Pp|,|\alpha|,|\beta|)$ from Lemma~\ref{lem:4} is at most six-fold exponential in $|\Pp|$.
Finally, taking into account Equality~\eqref{eq:def-phi}, where $\varphi$ is defined, we obtain a six-fold exponential bound for $\varphi$.

By Kučera and Mayr~\cite{KM10}, this yields a $6$-$\mathsf{EXPSPACE}$ upper bound for bisimulation finiteness of pushdown systems.
whereas only an exponential lower bound is known to the authors.
Determining the precise succinctness of pushdown systems with respect to finite systems
modulo bisimulation equivalence and determining the precise computational complexity of the bisimulation
finiteness problem for pushdown systems are natural candidates for a future work.

As mentioned in the introduction, it seems
worth investigating to which further classes of finitely-branching infinite-state systems our approach
can be applied, in particular among those for which the bisimulation equivalence problem
is not known to be decidable~\cite{Srba04}.

\begin{acks}
We thank Amina Doumane for evoking cooperation between the two authors.

Stefan G\"oller was supported by the \grantsponsor{ANRF}{Agence nationale de la recherche}{https://anr.fr/}
	(grant no. \grantnum[]{ANRF}{ANR-17-CE40-0010}).

Paweł Parys was supported by the \grantsponsor{NSCP}{National Science Centre, Poland}{https://www.ncn.gov.pl/} 
	(grant no. \grantnum[]{NSCP}{2016/22/E/ST6/00041}).
\end{acks}

\bibliography{bib}

\appendix

\newcommand{\Gg}{\mathcal{G}}
\newcommand{\DeltaP}{\Delta\!'}
\newcommand{\Suf}{\mathit{Suf}}
\newcommand{\unpack}{\overline}

\section*{Appendix}\setcounter{section}{1}

In the appendix we provide proofs omitted in the paper.
All these proofs are either easy or standard, but we include them for completeness.

\subsection*{Push-pop normal form}

In Lemma~\ref{lem:arbitrary2PushPop} we prove
that without loss of generality one can concentrate on pushdown systems in a push-pop normal form.

\begin{lemma}\label{lem:arbitrary2PushPop}
	For every PDS $\Pp$ one can compute in polynomial time a PDS $\Pp'$ in a push-pop normal form,
	such that for every initial configuration $s_0$ of $\Pp$, the quotients $\quotient{\Ll(\Pp,s_0)}$ and $\quotient{\Ll(\Pp',s_0)}$ are isomorphic.
\end{lemma}

\begin{proof}
	Let $\Pp=(Q,\Gamma,\Act,\Delta)$.
	We define a PDS $\Pp'=(Q,\Gamma',\Act,\allowbreak\DeltaP)$ in a push-pop normal form.

	Let $\Suf(\Delta)=\set{\alpha\in\Gamma^+\mid (p,X,a,q,\beta\alpha)\in\Delta}$ contain nonempty
	suffixes of stack contents pushed by transitions of $\Pp$.
	The new stack alphabet $\Gamma'=\Gamma\cup\Suf(\Delta)$ beside of these suffixes (treated now as single stack symbols) contains all stack symbols from $\Gamma$,
	so that an initial configuration of $\Pp$ is a valid configuration of $\Pp'$.
	We do not distinguish between elements of $\Gamma$ and suffixes of length one.

	For every transition $\delta=(p,X,a,q,\alpha)\in\Delta$, and for every string $X\!\beta\in\Gamma'$ we add to $\Delta'$ a transition
	\begin{itemize}
	\item	$(p,X\!\beta,a,q,(\alpha)(\beta))$ if $\alpha\neq\epsilon$ and $\beta\neq\epsilon$,
	\item	$(p,X\!\beta,a,q,(\alpha))$ if $\alpha\neq\epsilon$ and $\beta=\epsilon$,
	\item	$(p,X\!\beta,a,q,(\beta))$ if $\alpha=\epsilon$ and $\beta\neq\epsilon$, and
	\item	$(p,X\!\beta,a,q,\epsilon)$ if $\alpha=\epsilon$ and $\beta=\epsilon$.
	\end{itemize}

	For every stack content $\alpha$ over $\Gamma'$ there is a corresponding stack content $\unpack{\alpha}$ over $\Gamma$,
	obtained by concatenating strings from consecutive letters of $\alpha$.
	It is easy to see that
	\begin{itemize}
	\item	if $p\alpha\to_a q\beta$ in $\Pp'$, then $p\unpack{\alpha}\to_a q\unpack{\beta}$ in $\Pp$, and
	\item	if $p\unpack{\alpha}\to_a q\beta'$ in $\Pp$, then $p\alpha\to_a q\beta$ in $\Pp'$ for some stack content $\beta$ such that $\unpack{\beta}=\beta'$.
	\end{itemize}
	It follows that the quotients $\quotient{\Ll(\Pp,s_0)}$ and $\quotient{\Ll(\Pp',s_0)}$ are isomorphic for every initial configuration $s_0$ of $\Pp$.
\end{proof}

\subsection*{Proof of Lemma~\ref{lem:extend-eats}}

\restore\LemExtendEats

\begin{proof}
	Recall that $\eats{\beta}(q)$ contains control states $r$ such that $q\beta\to^*r$.
	For every such a control state $r$, due to $p\alpha\to^*q\beta$, we also have that $p\alpha\to^*r$, which implies that $r\in\eats{\alpha}(p)$.
\end{proof}

\subsection*{Proof of Lemma~\ref{lem:radius}}

\restore\LemRadius

\begin{proof}
	Fix the starting configuration $p\alpha$ and the number $z\in\Nat$.
	If $\alpha=\epsilon$, then $p\alpha$ has no successors, and the lemma holds trivially.
	Assume thus that $\alpha=X\!\alpha'$.

	Observe that if $p\alpha\xrightarrow{\rho}q\beta$ for some augmenting run $\rho$,
	then $\beta$ is of the form $\gamma\alpha'$ (the run never pops into $\alpha'$, it only builds some stack content on top of $\alpha'$).
	Moreover, the inequality $\sgrow(\rho)\leq z$ implies that $|\gamma|\leq z+1$.

	Thus, while choosing a configuration $q\beta=q\gamma\alpha'$ as in the statement of the lemma,
	we only need to choose a control state $q\in Q$, and a stack content $\gamma\in\Gamma$ of length at most $z+1$.
	There are $|Q|$ choices for the control state,
	and no more than $(|\Gamma|+1)^{z+1}$ choices for a stack content (every among $z+1$ positions either may be filled by a symbol from $\Gamma$, or may remain empty).
	Thus, the number of aforementioned configurations is at most $|Q|\cdot(|\Gamma|+1)^{z+1}\leq|\Pp|^{z+2}$.
	The inequality holds because $|Q|\leq|\Pp|$, and because $|\Gamma|+1\leq|\Gamma|+|Q|\leq|\Pp|$ if $|Q|>0$.
\end{proof}

\subsection*{Proof of Lemma~\ref{lem:PDS-paths}}

We start with an auxiliary lemma, which is analogous to Lemma~\ref{lem:PDS-paths}, but with the additional assumption that the stack content in the two considered configurations consists of a single symbol.

\begin{lemma}\label{lem:PDS-paths-aux}
	Let $p,q\in Q$ and $X,Y\in\Gamma$.
	If $pX\rightarrow^*qY$, then $\dist(pX,qY)\leq|\Pp|^{|\Pp|^4+1}-1$.
\end{lemma}

\begin{proof}
	Let $\rho$ be a shortest run from $pX$ to $qY$, and let $N=|\rho|$.
	Notice that $\rho$ never pops the bottommost stack symbol $X$ (after reaching an empty stack, the PDS gets stuck),
	that is, $\rho$ is augmenting.
	Denote $\Upsilon=|Q|^2\cdot|\Gamma|^2$.

	Suppose first that $\sgrow(\rho[0,i])\leq\Upsilon-1\leq|\Pp|^4-1$ for all $i\in[0,N]$.
	Then, by Lemma~\ref{lem:radius}, $\rho$ visits at most $|\Pp|^{|\Pp|^4+1}$ distinct configurations.
	Observe also that if $\rho(i)=\rho(j)$ for some $i,j\in[0,N]$ with $i<j$, then we can cut off the subrun $\rho[i,j]$, obtaining a shorter run from $pX$ to $qY$ contrarily to our assumption.
	Thus $\dist(pX,qY)=|[0,N]|-1\leq|\Pp|^{|\Pp|^4+1}-1$.

	Next, suppose that $\sgrow(\rho[0,m])\geq\Upsilon$ for some index $m\in[0,N]$.
	For every $j\in[0,\Upsilon]$ we define
	\begin{align*}
		\ell_j&=\max\set{i\in[0,m]\mid \sgrow(\rho[0,i])\leq j}
	\end{align*}
	and
	\begin{align*}
		r_j&=\min\set{i\in[m,N]\mid \sgrow(\rho[0,i])\leq j}\,.
	\end{align*}
	Notice that the positions $\ell_j,r_j$ are well-defined, that is, the sets are nonempty, because 
	\begin{align*}
		\sgrow(\rho[0,0])=\sgrow(\rho[0,N])=0\,.
	\end{align*}
	The stack growth changes by at most one in every step (since $\Pp$ is in a push-pop normal form),
	so 
	\begin{align*}
		\sgrow(\rho[0,\ell_j])=\sgrow(\rho[0,r_j])=j
	\end{align*}
	for every $j\in[0,\Upsilon]$.
	Moreover, $\sgrow(\rho[0,i])\geq j$ for all $i\in[\ell_j,r_j]$, which means that $\rho[\ell_j,r_j]$ is augmenting.

	For every $j\in[0,\Upsilon]$ let $\rho(\ell_j)=p_jX_j\alpha_j$ and $\rho(r_j)=q_jY_j\alpha_j$
	(because $\rho[\ell_j,r_j]$ is augmenting and because we have that $\sgrow(\rho[\ell_j,r_j])=0$, the stack content in the two configurations is the same, except for the topmost symbol).
	Recall that $|[0,\Upsilon]|=\Upsilon+1>|Q|^2\cdot|\Gamma|^2$.
	By the pigeonhole principle, there are two indices $j_1,j_2\in[0,\Upsilon]$ such that $j_1<j_2$ and $(p_{j_1},X_{j_1},q_{j_1},Y_{j_1})=(p_{j_2},X_{j_2},q_{j_2},Y_{j_2})$.
	The run $\rho[\ell_{j_2},r_{j_2}]$ can be seen as a run from $p_{j_2}X_{j_2}$ to $q_{j_2}Y_{j_2}$,
	thus also as a run from $p_{j_1}X_{j_1}\alpha_{j_1}$ to $q_{j_1}Y_{j_1}\alpha_{j_1}$.
	However $j_1<j_2$ implies $\ell_{j_1}<\ell_{j_2}$ and $r_{j_2}<r_{j_1}$.
	This contradicts the minimality of $\rho$, and thus finishes the proof: the subrun $\rho[\ell_{j_1},r_{j_1}]$ could be replaced by a shorter run $\rho[\ell_{j_2},r_{j_2}]$.
\end{proof}

\restore\LemPDSPaths

\begin{proof}
	We take $\Ee=|\Pp|^{|\Pp|^4+1}$.
	Let $\rho$ be a shortest run from $p\alpha$ to $q\beta$, and let $N=|\rho|$.
	Moreover, let
	\begin{align*}
		c=|\alpha|+\min\set{\sgrow(\rho[0,i])\mid i\in[0,N]}\,.
	\end{align*}
	We see that $c-1$ bottommost symbols of $\alpha$ are not touched by the run $\rho$,
	the $c$-th symbol is possibly modified but not popped, and the other symbols are popped, and then appropriate symbols of $\beta$ are pushed.
	Furthermore, let
	\begin{align*}
		\ell_j&=\min\set{i\in[0,N]\mid|\alpha|+\sgrow(\rho[0,i])\leq j}
	\end{align*}
	for all $j\in[c,|\alpha|]$, and
	\begin{align*}
		r_j&=\max\set{i\in[0,N]\mid|\alpha|+\sgrow(\rho[0,i])\leq j}
	\end{align*}
	for all $j\in[c,|\beta|]$, where $\ell_{|\alpha|}=0$ and $r_{|\beta|}=N$.
	Because the stack growth changes by one ($\Pp$ is in a push-pop normal form),
	it is easy to see that
	\begin{itemize}
	\item	$\rho[\ell_j,\ell_{j-1}-1]$ for $j\in[c+1,|\alpha|]$ is an augmenting run with $\sgrow(\rho[\ell_j,\allowbreak\ell_{j-1}-1])=0$, so $\ell_{j-1}-1-\ell_j\leq\Ee-1$ by Lemma~\ref{lem:PDS-paths-aux};
	\item	$\rho[r_{j-1}+1,r_j]$ for $j\in[c+1,|\beta|]$ is an augmenting run with $\sgrow(\rho[r_{j-1}+1,\allowbreak r_j])=0$, so $r_j-(r_{j-1}+1)\leq\Ee-1$ by Lemma~\ref{lem:PDS-paths-aux};
	\item	if $c\geq 1$, then $\rho[\ell_c,r_c]$ is an augmenting run with $\sgrow(\rho[\ell_c,r_c])=0$, so $r_c-\ell_c\leq\Ee-1$ by Lemma~\ref{lem:PDS-paths-aux};
	\item	if $c=0$, then $\ell_c=N$ (and $|\beta|=0$).
	\end{itemize}
	Summing this up, for $c\geq 1$ we obtain that $N\leq(|\alpha|-c)\cdot\Ee+(|\beta|-c)\cdot\Ee+\Ee-1\leq(|\alpha|+|\beta|)\cdot\Ee$,
	and for $c=0$ we obtain that $N\leq|\alpha|\cdot\Ee\leq(|\alpha|+|\beta|)\cdot\Ee$.
	This finishes the proof.
\end{proof}

\subsection*{Proof of Lemma~\ref{lem:decompose}}

In this proof we use a notion of semigroups.
A \emph{semigroup} $(S,\circ)$ is a set $S$ equipped with a binary operation $(\circ)\colon S\times S\to S$ that is associative (i.e., $s\circ(t\circ u)=(s\circ t)\circ u$ for all $s,t,u\in S$).
An \emph{idempotent} is an element of $S$ such that $s\circ s=s$.

\begin{lemma}\label{lem:decompose-aux}
	Let $(S,\circ)$ be a semigroup, and let $s_1\dots s_n\in S^*$ be a word such that $n\geq 2^{9\cdot|S|}$.
	Then there exist indices $j,k$ such that $1\leq j<k\leq2^{9\cdot|S|}$, and $s_{j+1}\circ\dots\circ s_k$ is an idempotent, and $s_1\circ\dots\circ s_j=s_1\circ\dots\circ s_k$.
\end{lemma}

\begin{proof}
	We prove the lemma assuming that $n=2^{9\cdot|S|}+1$.
	A general situation can be reduced to this situation by truncating the word, or by adding an arbitrary symbol at the end of the word;
	such changes do not influence the thesis, which talks only about positions up to $2^{9\cdot|S|}$.

	A \emph{factorization tree} over a word $s_1\dots s_n\in S^+$ is an $S$-labeled tree with $n$ leaves, where the $i$-th leaf is labeled by $s_i$, for all $i\in[1,n]$,
	and where every internal node is labeled by the product of labels of its children.
	Notice that a node $v$ having leaves number $i+1,i+2,\dots,j$ as its descendant is labeled by a product $s_{i+1}\circ\dots\circ s_j$.

	A factorization tree is \emph{ramseyan} if every its node either
	\begin{itemize}
	\item	is a leaf,
	\item	has two children, or
	\item	its children are labeled by the same \emph{idempotent}, that is, an element $s\in S$ such that $s\circ s=s$.
	\end{itemize}

	A theorem by Simon~\cite{Simon} says that for every nonempty word $s_1\dots s_n\in S^+$ there exists a ramseyan factorization tree $T$ over $s_1\dots s_n$ of height at most $9\cdot|S|$
	(where height of a tree is defined as the maximal number of edges on a path from the root to a leaf).

	Because $n>2^{9\cdot|S|}$, necessarily $T$ has a node with more than two children; let $v$ be such a node,
	and let $v_1,v_2$ be its first two children.
	Because $T$ is ramseyan, $v_1$ and $v_2$ are labeled by the same idempotent $s$.
	Note that $s$, as the label of $v_1$, equals $s_{i+1}\circ\dots\circ s_j$ and, as the label of $v_2$, equals $s_{j+1}\circ\dots\circ s_k$, for appropriate indices such that $0\leq i<j<k<n$
	(we have $k<n$, not $k\leq n$, because $v_2$ is not the last child of its parent).
	Observe also that
	\begin{align*}
		s_1\circ\dots\circ s_j&=(s_1\circ\dots\circ s_i)\circ s\\
			&= (s_1\circ\dots\circ s_i)\circ s\circ s\\
			&=s_1\circ\dots\circ s_k\,.
	\end{align*}
	Thus, $j$ and $k$ satisfy the statement of the lemma.
\end{proof}

Next, in Lemma~\ref{lem:decompose-RP}, we show that a linked pair can be found on top of every large enough stack content,
without caring about reachability from the initial configuration.

\begin{lemma}\label{lem:decompose-RP}
	There is a constant $\Gg\leq\exp(\exp(|\Pp|))$ such that
	for every sequence of nonempty stack contents $\alpha_1,\dots,\alpha_n$ with $n\geq\Gg$ there exist indices $j,k$ such that $1\leq j<k\leq\Gg$
	and $(\alpha_1\dots\alpha_j,\alpha_{j+1}\dots\alpha_k)$ is a linked pair.
\end{lemma}

\begin{proof}
	Let $S$ be the set of binary relations over $Q$, and let $\circ$ be the operation of relation composition.
	Clearly $\circ$ is associative, and thus $(S,\circ)$ is a semigroup.
	Moreover, for a stack content $\alpha$ let \begin{align*}
		f(\alpha)=\set{(p,q)\in Q\times Q\mid q\in\eats{\alpha}(p)}\,.
	\end{align*}
	Because $\eats{\beta}(\eats{\alpha}(p))=\eats{\alpha\beta}(p)$, we have that $f(\alpha)\circ f(\beta)=f(\alpha\beta)$ for all stack contents $\alpha,\beta$.
	Notice also that $\eats{\alpha}(T)=\set{q\mid\exists p\in T.\, (p,q)\in f(\alpha)}$ for every set of control states $T\subseteq Q$;
	in particular $f(\alpha)$ determines $\eats{\alpha}$.
	We take
	\begin{align*}
		\Gg=2^{9\cdot 2^{|\Pp|^2}}\geq 2^{9\cdot 2^{|Q|^2}}=2^{9\cdot|S|}\,.
	\end{align*}

	Consider now a sequence of stack contents $\alpha_1,\dots,\alpha_n$ with $n\geq\Gg$.
	For every $i\in[1,n]$, let $s_i=f(\alpha_i)$.
	By Lemma~\ref{lem:decompose-aux}, there are indices $j,k$ such that $1\leq j<k\leq2^{9\cdot|S|}\leq\Gg$, and $s_{j+1}\circ\dots\circ s_k$ is an idempotent, and $s_1\circ\dots\circ s_j=s_1\circ\dots\circ s_k$.
	We have that
	\begin{align*}
		f(\alpha_{j+1}\dots\alpha_k)&=s_{j+1}\circ\dots\circ s_k\\
			&=(s_{j+1}\circ\dots\circ s_k)\circ(s_{j+1}\circ\dots\circ s_k)&\\
			&=f(\alpha_{j+1}\dots\alpha_k\alpha_{j+1}\dots\alpha_k)\,,
	\end{align*}
	meaning that $\eats{\alpha_{j+1}\dots\alpha_k}=\eats{\alpha_{j+1}\dots\alpha_k\alpha_{j+1}\dots\alpha_k}$.
	Likewise,
	\begin{align*}
		f(\alpha_1\dots\alpha_j)=s_1\circ\dots\circ s_j=s_1\circ\dots\circ s_k=f(\alpha_1\dots\alpha_k)\,,
	\end{align*}
	meaning that $\eats{\alpha_1\dots\alpha_j}=\eats{\alpha_1\dots\alpha_k}$.
	The two equalities imply that $(\alpha_1\dots\alpha_j,\alpha_{j+1}\dots\alpha_k)$ is a linked pair.
\end{proof}

\restore\LemDecompose

\begin{proof}
	We take $\Ff=\Gg\cdot|\Pp|^2$, where $\Gg$ is the constant from Lemma~\ref{lem:decompose-RP}.

	Let $\rho$ be a run from $q_0\alpha_0$ to a configuration $q\delta$ such that $|\delta|\geq\Ff+|\alpha_0|$.
	Denote $n=|\delta|$, $m=n-|\alpha_0|$, and $\delta=X_1\dots X_n$.
	Observe that $\rho$ can be split into subruns $\rho_{m+1},\rho_m,\dots,\rho_1$,
	where the runs $\rho_i$ for $i\in[1,m]$ are augmenting and such that $\sgrow(\rho_i)=1$
	(the run $\rho_{m+1}$ needs not to be augmenting, and satisfies $\sgrow(\rho_{m+1})=0$).
	To this end, for $d=m,m-1,\dots,1$ we split the run on the last moment when $X_{d+1}$ is the topmost stack symbol.
	More formally, for $d\in[0,m]$, we take
	\begin{align*}
		i_d=\max\set{i\in[0,|\rho|]\mid\sgrow(\rho[i,|\rho|])\geq d}\,,
	\end{align*}
	and we define $\rho_{m+1}=\rho[0,i_m]$ and $\rho_d=\rho[i_d,i_{d-1}]$ for $d\in[1,m]$ (notice that $i_0=|\rho|$).
	Then, for $d\in[1,m+1]$, let $q_d=\rho(i_{d-1})$ be the control state in which $\rho_d$ ends.

	Because $\Ff+1\geq\Gg\cdot|\Pp|^2+1\geq\Gg\cdot|Q|\cdot|\Gamma|+1$, by the pigeonhole principle there exist
	$\Gg+1$ indices $1\leq\ell_0<\ell_1<\dots<\ell_\Gg\leq\Ff+1$ such that $(q_{\ell_i},X_{\ell_i})=(q_{\ell_j},X_{\ell_j})$ for all $i,j\in[0,\Gg]$.
	Recall that $\Ff+1\leq m+1$, by assumption.
	Let us split the stack $\delta$ at these indices: let
	\begin{align*}
		\alpha_0&=X_1\dots X_{\ell_0-1},\\
		\alpha_i&=X_{\ell_{i-1}}\dots X_{\ell_i-1}\qquad\mbox{for all $i\in[1,\Gg]$, and}\\
		\alpha_{\Gg+1}&=X_{\ell_\Gg}\dots X_m\,.
	\end{align*}
	Notice that $\alpha_i$ for $i\in[1,\Gg]$ are nonempty.

	Next, we apply Lemma~\ref{lem:decompose-RP} to the sequence $\alpha_1,\dots,\alpha_\Gg$ (thus, without $\alpha_0$ and $\alpha_{\Gg+1}$).
	It gives us indices $j,k$ such that $1\leq j<k\leq\Gg$ and $(\alpha_1\dots\alpha_j,\alpha_{j+1}\dots\alpha_k)$ is a linked pair.
	We take $\alpha=\alpha_0\alpha_1\dots\alpha_j$, and $\beta=\alpha_{j+1}\dots\alpha_k$, and $\gamma=\alpha_{k+1}\dots\alpha_{\Gg+1}$.
	By definition $\delta=\alpha\beta\gamma$ and $|\alpha|+|\beta|=\ell_k-1\leq\Ff$.
	Moreover $(\alpha,\beta)$ is a linked pair (notice that $\eats{\alpha_1\dots\alpha_j}=\eats{\alpha_1\dots\alpha_k}$ implies $\eats{\alpha_0\alpha_1\dots\alpha_j}=\eats{\alpha_0\alpha_1\dots\alpha_k}$).

	It remains to prove that all configurations of the form $q\alpha\beta^i\gamma$ (with $i\in\Nat$) are reachable from $q_0\alpha_0$.
	Take some $i\in\Nat$.
	We can reach $q\alpha\beta^i\gamma$ as follows.
	First, starting from $q_0\alpha_0$, we use the composition of the subruns $\rho_{m+1},\rho_m,\dots\rho_{\ell_k}$ to reach $q_{\ell_k}\gamma$.
	Next, we repeat $i$ times the composition of the subruns $\rho_{\ell_k-1},\rho_{\ell_k-2},\dots,\rho_{\ell_j}$ to reach $q_{\ell_j}\beta^i\gamma$.
	Finally, we use the composition of the subruns $\rho_{\ell_j-1},\rho_{\ell_j-2},\dots\rho_1$ to reach $q\alpha\beta^i\gamma$.
	Here it is important that $(q_{\ell_j},X_{\ell_j})=(q_{\ell_k},X_{\ell_k})$; recall that $\rho_{\ell_k-1}$ starts in the control state $q_{\ell_k}$ while $\rho_{\ell_j}$ ends in the control state $q_{\ell_j}$,
	and that the topmost symbol of $\gamma$ is $X_{\ell_k}$ while the topmost symbol of $\beta$ is $X_{\ell_j}$.
	It is also important that the subruns $\rho_{\ell_k-1},\rho_{\ell_k-2},\dots,\rho_1$ are augmenting.
	These facts imply that the middle fragment, creating $\beta$, can be repeated as many times as we want, and then it can be followed by the final fragment, creating $\gamma$.
\end{proof}

\begin{remark*}
	While proving Lemma~\ref{lem:decompose-aux} we use Simon's theorem, which is a quite powerful tool.
	Nevertheless, the authors are not aware of any ``trivial'' proof of Lemma~\ref{lem:decompose-aux}.
	It seems that while proving this lemma directly, it is anyway necessary to use algebraic arguments involving the theory of semigroup ideals (the relations of Green),
	like in the proof of the Simon's theorem.
	In particular, a proof of Lemma~\ref{lem:decompose-RP} using Ramsey's theorem, without referring at all to the semigroup structure, gives us a non-elementary upper bound.
\end{remark*}

\subsection*{Proof of Lemma~\ref{lem:eats-beta}}

\restore\LemEatsBeta

\begin{proof}
	This follows directly from the definitions:
	\begin{align*}
		\eats{\beta}(r')\subseteq\eats{\beta}(\eats{\alpha}(r))=\eats{\alpha\beta}(r)=\eats{\alpha}(r)\,.
	\tag*{\qedhere}
	\end{align*}
\end{proof}

\subsection*{Proof of Lemma~\ref{lem:equiv-because-equiv-below}}

\restore\LemEquivBecauseEquivBelow

\begin{proof}
	Fix $q,\alpha,\gamma,\gamma'$ as in the statement of the lemma.
	We define
	\begin{align*}
		R=(\sim)\cup\set{(p\delta\gamma,p\delta\gamma')\mid p\in Q, \delta\in\Gamma^*, q\alpha\to^* p\delta}\,.
	\end{align*}
	Observe that in particular $(q\alpha\gamma,q\alpha\gamma')\in R$, thus it is enough to prove that $R$ is a bisimulation.

	To this end, consider a pair $(s,s')\in R$, and assume that $s\to_at$.
	We should prove that there exists some $t'$ such that $s'\to_at'$ and $(t,t')\in R$.
	If $(s,s')\in(\sim)$, the existence of an appropriate $t'$ follows from the fact that $\sim$ is a bisimulation.
	Thus, assume that $s$ and $s'$ are of the form $p\delta\gamma$ and $p\delta\gamma'$, respectively, where $q\alpha\to^* p\delta$.
	If $\delta=\epsilon$, then $p\in\eats{\alpha}(q)$,
	so we actually have $s=p\gamma\sim p\gamma'=s'$ by assumptions of the lemma.
	This means that the case of $\delta=\epsilon$ is already covered by the case of $(s,s')\in(\sim)$.
	It remains to consider the situation when $\delta=X\eta$ for some $X\in \Gamma$ and some $\eta\in\Gamma^*$.
	Because $pX\eta\gamma=s\to_a t$, necessarily $t=r\beta\eta\gamma$ for some transition $(p,X,a,r,\beta)\in\Delta$.
	Take $t'=r\beta\eta\gamma'$.
	Due to the same transition we have that $s'=pX\eta\gamma'\to_a t'$.
	Moreover, $q\alpha\to^* p\delta\to^* r\beta\eta$,
	and in consequence $(t,t')\in R$.

	We should now also consider $t'$ such that $s'\to_at'$, and prove that there exists some $t$ such that $s\to_at$ and $(t,t')\in R$.
	This is completely symmetric to what we have just done.
\end{proof}

\subsection*{Proof of Lemma~\ref{lem:repetition-is-omega}}

\restore\LemRepetitionIsOmega

\begin{proof}
	Fix $U,\beta,i,j$ as in the statement of the lemma, and assume (without loss of generality) that $i<j$.
	We define
	\begin{align*}
		R=\set{(s,p\delta\beta^\omega)\mid s&\in Q\times\Gamma^*, p\in Q, \delta\in\Gamma^*, \\
			s&\sim p\delta\beta^i\gamma,\exists r\in U.\,r\beta^{j-i}\to^* p\delta}\,.
	\end{align*}
	Observe that in particular $(r\beta^j\gamma,r\beta^\omega)\in R$ for all $r\in U$ (we write $r\beta^j\gamma$ as $r\delta\beta^i\gamma$ for $\delta=\beta^{j-i}$,
	which implies the thesis if $R$ is a bisimulation.

	In order to prove that $R$ is a bisimulation, consider a pair $(s,p\delta\beta^\omega)\in R$.
	By definition $s\sim p\delta\beta^i\gamma$ and $r\beta^{j-i}\to^* p\delta$ for some control state $r\in U$.

	Assume first that $|\delta|>0$, that is, $\delta=X\eta$ for some $X\in\Gamma, \eta\in\Gamma^*$.
	Suppose that $s\to_a t$ for some configuration $t$; we should prove
	the existence of a configuration $t''$ such that $pX\eta\beta^\omega\to_a t''$ and $t\sim t''$.
	Because $s\sim pX\eta\beta^i\gamma$, there is $t'$ such that $pX\eta\beta^i\gamma\to_a t'$ and $t\sim t'$.
	Necessarily $t'=q\alpha\eta\beta^i\gamma$ for some transition $(p,X,a,q,\alpha)\in\Delta$.
	Take $t''=q\alpha\eta\beta^\omega$.
	Due to the same transition we have that $pX\eta\beta^\omega\to_a t''$.
	Moreover, $r\beta^{j-i}\to^* pX\eta\to q\alpha\eta$, so $(t,t'')\in R$.

	Conversely, suppose that $pX\eta\beta^\omega\to_a t''$ for some configuration $t''$; we should prove the existence of a configuration $t$ such that $s\to_a t$ and $t\sim t''$.
	Necessarily $t''=q\alpha\eta\beta^\omega$ for some transition $(p,X,a,q,\alpha)\in\Delta$.
	Due to the same transition we have that $pX\eta\beta^i\gamma\to_a q\alpha\eta\beta^i\gamma$,
	and because $s\sim pX\eta\beta^i\gamma$, there is a configuration $t$ such that $s\to_at$ and $t\sim q\alpha\eta\beta^i\gamma$.
	Again, $r\beta^{j-i}\to^* pX\eta\to q\alpha\eta$, so $(t,t'')\in R$.

	It remains to consider the case when $\delta=\epsilon$.
	Because $\eats{\beta}(\cdot)$ is monotone, $\eats{\beta}(U)\subseteq U$ implies $\eats{\beta}(\eats{\beta}(U))\subseteq \eats{\beta}(U)$, and in effect $\eats{\beta^k}(U)\subseteq U$ for all $k\geq 1$.
	Because $r\beta^{j-i}\to^* p\delta=p$, we have that $p\in\eats{\beta^{j-i}}(U)\subseteq U$.
	Thus, by assumption of the lemma, we have that $p\beta^i\gamma\sim p\beta^j\gamma$.
	Recall also that $s\sim p\beta^i\gamma$.
	Taking $\delta'=\beta^{j-i}$, we have that $p\delta\beta^\omega=p\delta'\beta^\omega$, and $s\sim p\delta'\beta^i\gamma$, and $p\beta^{j-i}\to^* p\delta'$.
	This means that the case of $\delta=\epsilon$ is already covered by the previous case (where we take $\delta'$ instead of $\delta$).
\end{proof}

\subsection*{Proof of Corollary~\ref{cor:repetition-is-omega}}

\restore\CorRepetitionIsOmega

\begin{proof}
	For $U=\eats{\alpha}(q)$ we see that $\eats{\beta}(U)=U$ (because $(\alpha,\beta)$ is a linked pair).
	In effect, Lemma~\ref{lem:repetition-is-omega} implies that $r\beta^i\gamma\sim r\beta^j\gamma\sim r\beta^\omega$ for all $r\in\eats{\alpha}(q)$,
	and thus $q\alpha\beta^i\gamma\sim q\alpha\beta^j\gamma\sim q\alpha\beta^\omega$ by Lemma~\ref{lem:equiv-because-equiv-below}.
\end{proof}

\end{document}